\documentclass[11pt,letterpaper]{article}
\usepackage{setspace}
\usepackage[margin=1in, paper=letterpaper]{geometry}

\usepackage[utf8]{inputenc}
\usepackage{amsmath,amssymb,amsfonts}
\usepackage{amsthm}
\usepackage{hyperref}
\usepackage{tikz}

\usepackage{xspace,xcolor}
\usepackage{delimset}
\usepackage{thmtools}
\usepackage{cleveref}
\usepackage{comment}
\usepackage{pifont}

\usepackage{natbib}
\usepackage[inline]{enumitem}
\setlist[enumerate]{label=(\arabic*)}

\newtheorem{theorem}{Theorem}
\newtheorem{lemma}{Lemma}
\newtheorem{corollary}{Corollary}
\newtheorem{definition}{Definition}

\newtheorem{observation}{Observation}

\newcommand{\size}{\delim{|}{|}}
\newcommand{\ceil}{\delim{\lceil}{\rceil}}
\newcommand{\floor}{\delim{\lfloor}{\rfloor}}
\newcommand{\set}{\delimpair{\{}{[m]|}{\}}}
\newcommand{\brc}{\brk[c]}
\newcommand{\tpl}{\brk[a]}
\newcommand{\bitf}{\delimpair{[}{[i]|}{]}}
\renewcommand{\mod}{\,\mathrm{mod}\,}

\newcommand*{\AC}{{AC\textsuperscript{0}}\xspace}
\let\tty\texttt
\let\eps\epsilon
\def\Th{{\it \Theta}}
\DeclareMathOperator{\DD}{\mathcal{D}}

\newcommand*{\RANK}{\textup{\tt rank}\xspace}
\newcommand*{\SELECT}{\texttt{\tt select}\xspace}
\newcommand*{\SUM}{\texttt{\tt sum}\xspace}
\newcommand*{\SEARCH}{\texttt{\tt search}\xspace}
\newcommand*{\MSB}{\texttt{\tt msb}\xspace}
\newcommand*{\LSB}{\texttt{\tt lsb}\xspace}
\newcommand*{\WT}{\texttt{\tt wt}\xspace}
\newcommand*{\MATCH}{\texttt{\tt match}\xspace}

\newcommand*{\OPT}{\textup{\tt OPT}\xspace}

\DeclareMathOperator{\polylog}{polylog}

\renewcommand*{\O}{\mathcal{O}}

\newcommand*{\checkyes}{\text{\ding{51}}}
\newcommand*{\checkno}{\text{\ding{55}}}
\newcommand*{\checkdagger}{\text{\checkyes\textsuperscript{$\dagger$}}\xspace}

\newcommand*{\Pack}[1]{\Pi_{#1}}
\newcommand*{\sumstyle}{\displaystyle}

\begin{document}

\begin{titlepage}

\title{Compressing Dynamic Fully Indexable Dictionaries in Word-RAM}
\author{Gabriel Marques Domingues\footnote{
{\tt {gm@mail.tau.ac.il}}.
School of Electrical Engineering, Tel Aviv University.
This research was supported by the Israel Science Foundation, grant No. 1948/21.}}
\date{}

\maketitle


\begin{abstract}
We study the problem of constructing a dynamic fully indexable dictionary (FID) in the Word-RAM model using space close to the information-theoretic lower bound.
A FID is a data-structure that encodes a bit-vector $B$ of length $u$ and answers, for $b\in\{0,1\}$,
$\texttt{rank}_b(B, x)=|{\{y\leq x~|~B[y]=b\}}|$ and $\texttt{select}_b(B, r)=\min\{0\leq x<u~|~\texttt{rank}_b(B, x)=r\}$ ($-1$ if empty).
A dynamic FID supports updates that modify a single bit of $B$, i.e., $B[i]\gets b$.

We work in the Word-RAM model with $w$-bit words, assuming $w\geq \operatorname{lg} u$.
Integer multiplication takes $\mathcal{O}(1)$ time.
Our memory model is $\mathcal{M}_B$,
allowing access to a fixed precomputed table of $\tau=\operatorname{polylog}(w)$ words, which can be computed in $\mathcal{O}(w\tau)$ time.

In this paper, we show a dynamic FID
based on the famous fusion-tree data-structure of
P{\u{a}}tra{\c{s}}cu and Thorup [FOCS 2014],
modified to use fewer bits and to support $\texttt{select}_0$.
Let $n$ denote the number of ones in $B$.
We describe a parametric construction:
for every $\epsilon\leq 1/2$,
there is a dynamic FID using $$\operatorname{lg}\binom{u}{n}+\mathcal{O}(nw^{\epsilon}/\epsilon)\text{ bits}$$
taking $\mathcal{O}({1/\epsilon+\log_w(n)})$ time for $\texttt{rank}_0/\texttt{rank}_1/\texttt{select}_0$ and updates,
and $\mathcal{O}({\log_w(n)})$ time for $\texttt{select}_1$.
All time bounds are worst-case.
For $\epsilon={1/\sqrt{\operatorname{lg} w}}$,
we reduce the space to $\operatorname{lg}\binom{u}{n}+\mathcal{O}(n\log w)$ bits.
For $\epsilon=\Theta(1)$, the running time matches the lower bound of
Fredman and Saks [STOC 1989].
This is the first deterministic dynamic FID in the standard Word-RAM model
that achieves $o(n\sqrt{w})$ bits of redundancy in $\mathcal{M}_B$ (e.g., $\epsilon=1/4$),
and optimal worst-case time.
\end{abstract}

\thispagestyle{empty}
\end{titlepage}


\section{Introduction}

In this paper, we consider deterministic data-structures in the Word-RAM model with $w$-bit words.
Supporting
$\RANK/\SELECT$ in Word-RAM is a widely studied problem~\cite{jacobson1988succinct,clark1996efficient,clark1996compact,raman2007succinct,golynski2007optimal,golynski2007size,patrascu2008succincter,grossi2009more,golynski2014optimal,li2023dynamic,liang2025optimal,kuszmaul2026succinct},
especially for implementing predecessor search \cite{fredman1990blasting,fredman1993surpassing,andersson1999fusion,andersson2007dynamic,patrascu2014dynamic,cohen2015minimal,navarro2020predecessor}.

A \emph{fully indexable dictionary} (FID) is a data-structure that encodes a bit-vector $B\in\brc{0,1}^u$ and
supports the following operations~\cite{raman2007succinct,golynski2007size,golynski2008redundancy,grossi2009more,golynski2014optimal,navarro2014fully}:
\begin{enumerate}
\item $\RANK_b(B, x)=\size{\set{y\leq x}{B[y]=b}}$, and
\item $\SELECT_b(B, r)=\min\set{0\leq x < u}{\RANK_b(B, x)=r}$ ($-1$ if empty).
\end{enumerate}
A dynamic FID supports writes to $B$, that is,
updating a bit at position $i$ with value $b$ ($B[i]\gets b$).
Note that $\RANK_0(B, i)=i+1-\RANK_1(B, i)$; from now on, we shall use $\RANK$ interchangeably with $\RANK_1$.
In general, implementing $\SELECT_0$ based on $\RANK$ and $\SELECT_1$ in the dynamic setting is non-trivial, as it usually requires more space~\cite{grossi2009more,navarro2014fully,navarro2014optimal}.
See~\cite{raman2007succinct,navarro2014fully} for applications of $\SELECT_0$.\footnote{%
~\cite{patrascu2008succincter,patrascu2014dynamic,li2023dynamic,liang2025optimal,kuszmaul2026succinct} do not address $\SELECT_0$ explicitly.
}

An equivalent description of a FID is to consider the set $S=\set{0\leq x<u}{B[x]=1}$.
In this view, a FID encodes a set $S$ of $n$ elements in $[u]=\brc{0,1,\cdots, u-1}$,
supporting $\RANK(S, x)=\size{\set{y\in S}{y\leq x}}$ and $\SELECT_0/\SELECT_1$ operations as before.
A dynamic FID supports insertions into $S$ and deletions from $S$.
Membership and predecessor queries to $S$ can be answered by computing $\SELECT_1(S, \RANK(S, x))$.
For further details on predecessor data-structures,
see the survey of~\cite{navarro2020predecessor}.

We focus on the density-sensitive case,
where the data-structure
knows both $u$ and $n=\size{S}$ (i.e., the number of ones of $B$)~\cite{grossi2009more,golynski2014optimal}.
We seek a FID that uses space close to the information-theoretic lower bound to encode $B$ (equivalently $S$), which is given by $\lg\binom{u}{n}$ bits.
For the running time,
a lower bound
is given in~\cite{fredman1989cell,patrascu2014dynamic}
and states that
at least one of $\RANK_1/\SELECT_1/$update
must take $\Omega(\log_w(n))$ time.
The first deterministic dynamic structure with worst-case $\O(\log_w(n))$ time
using $\O(nw)$ space
is the dynamic fusion tree of~\cite{patrascu2014dynamic}.

For dynamic memory allocation, we use the $\mathcal{M}_B$ memory model \cite{raman2003succinct,arroyuelo2016succinct,li2023dynamic}.
In this model, the space occupied by a dynamic data-structure equals
the highest address in memory that it uses.
Dynamically maintaining space near $\lg\binom{u}{n}$ is generally difficult in this model~\cite{li2023dynamic}.
Specifically, if two dynamic data-structures using $s_1$
and $s_2$ bits, respectively,
are concatenated,
the resulting space may exceed $s_1+s_2$ bits
after multiple updates~\cite{li2023dynamic}.
Standard techniques for concatenating dynamic memory segments \cite{raman2003succinct,navarro2014optimal,arroyuelo2016succinct,li2023dynamic,kuszmaul2026succinct}
require either a fixed number of memory segments based on $u$~\cite{li2023dynamic},
or a rebuild whenever $n$ changes significantly~\cite{kuszmaul2026succinct}.

Designing a space-efficient dynamic FID with $\O(\log_w n)$ worst-case time operations has been a challenge.
Apart from~\cite{patrascu2014dynamic},
most deterministic structures are either amortized,
require large precomputed tables,
or rely on periodically rebuilding of static indices~\cite{navarro2014fully,navarro2020predecessor,li2023dynamic,navarro2025worst}.
Our main result is a deterministic dynamic FID that uses
$\lg\binom{u}{n}+\O(nw^{\eps})$ space (for any constant $\eps<1$)
and has worst-case $\O(\log_w n)$ time operations,
while only using $\O(\polylog w)$ precomputed words.
See~\Cref{tbl:results} for a list of recent relevant results
and the results presented in this paper.

\begin{table*}[t]
\centering
\renewcommand{\arraystretch}{1}
\begin{tabular}{|c|c|c|c|c|}\hline
Reference & Dyn. & Redundancy & Table & Time \\\hline
\cite{patrascu2014dynamic} & \checkyes & $\O(nw)$ & $\polylog(w)$ & $\O(\log_w(n))$ \\\hline
\cite{navarro2014fully} & \checkyes & $\O(u\log\log u/\log u)$ & $\sqrt{u}$ & $\O\brk*{\frac{\log u}{\log\log u}}$ \\\hline
\cite{li2023dynamic}& \checkyes & $\O\brk*{{u\sqrt{w}}/{2^{(\log u/\log\log u)^{1/5}}}}$ & ${2^{\polylog(u)}}$ & $\O\brk*{\frac{\log u}{\log\log u}}$ \\\hline
\cite{navarro2025worst}& \checkdagger  & $\O\brk*{n\log\log(u/n)}$ & $\sqrt{u}$ & $\O\brk*{\frac{\log n}{\log\log n}}$ \\\hline
\cite{kuszmaul2026succinct}& \checkdagger  & $\O\brk*{{nw}/{2^{(\log_w n)^{1/3}}}}$ & ${\polylog(u)}$ & $\O_{\text{exp}}(\log_w(n))$  \\\hline
\Cref{lemma:memory}+\cite{kuszmaul2026succinct}& \checkyes  & $\O\brk*{n+{nw}/{2^{(\log_w(n))^{1/3}}}}$ & ${\polylog(u)}$ & $\O_{\text{exp}}(\log_w(n))$  \\\hline
Ours (\Cref{thm:dynamic}) & \checkyes & $\O(nw^{\eps}/\eps)$ & ${\polylog(w)}$ & $\O(1/\eps+\log_w(n))$ \\\hline
Ours (\Cref{thm:sqrt-log}) & \checkyes & $\O(n\log w)$ & ${\polylog(w)}$ & $\O\brk*{\sqrt{\log w}+\log_w(n)}$ \\\hline
\end{tabular}
\caption{
Survey of recent results on dynamic FIDs and $\RANK/\SELECT_1$ data-structures.
Rows marked ``\checkyes'' for ``Dyn.'' are dynamic,
and marked ``\checkdagger'' have amortized updates.
A FID with redundancy $r$ uses $\lg\binom{u}{n}+r$ bits,
seen in the ``Redundancy'' column.
The ``Table'' column is the space used for lookup tables.
The ``Time'' column is for $\RANK$ and updates.
Note that~\cite{kuszmaul2026succinct} is a randomized data-structure;
running times in expectation are denoted by $\O_{\text{exp}}$.
}
\label{tbl:results}
\end{table*}


\subsection{Contributions}\label{sec:results}
In this section,
we list our contributions and give an overview of our techniques.
Recall that a dynamic FID encodes a set $S$ of $n$ elements of $[u]=\brc{0,1, \cdots,u-1}$ and supports $\RANK/\SELECT$ and updates.
We employ the standard Word-RAM model
with $w$-bit words and $\O(1)$ time multiplication,
and use the $\mathcal{M}_B$ memory model, expressing the space in bits.
We assume $w\geq \lg u$.
All results use a precomputed table of $\polylog(w)$ words.
A FID is said to have redundancy $r$ if it uses $\lg\binom{u}{n}+r$ bits,
for all $u$ and $n$ (the redundancy depends on $u$ and $n$).
All running times are worst-case and all data-structures are deterministic.
We prove the following theorem.

\begin{theorem}[Dynamic FID]\label{thm:dynamic}
For every $\eps\leq 1/2$,
there is a dynamic FID
with $\O(nw^{\eps}/\eps)$ redundancy
that takes $\O\brk*{1/\eps+\log_w(n)}$ time for $\RANK/\SELECT_0/$updates
and $\O\brk*{\log_w(n)}$ time for $\SELECT_1$.
\end{theorem}

\medskip\noindent
We divide the proof of this result into two parts:
\begin{enumerate*}
\item reducing the problem to a small FID, and
\item implementing a small FID by constructing separate data-structures for each operation.
\end{enumerate*}

\subsubsection{Cardinality reduction}

We prove the following general lemma for managing memory in $\mathcal{M}_B$
to reduce the problem to a FID for small set.
The idea is to construct a tree with degree $w^{\O(1)}$,
with FIDs at the leaves.
Observe that $n$ is not known a priori.
This lemma replaces previous memory management tricks for FIDs~\cite{navarro2014optimal,li2023dynamic,kuszmaul2026succinct}, allowing for a dynamic number of leaves and with worst-case time guarantees.

\begin{restatable}{lemma}{LemmaMemory}\label{lemma:memory}
For any $N\geq 1$ and $r=\O(Nw)$,
suppose that there is a dynamic FID for a set $\O(N)$ elements using
$\O(r)$ redundancy.
Then, for any $n$, there is a dynamic FID for a set of $n$ elements
with $\O\brk{n+nr/N+nw/\sqrt{N}}$ redundancy
and the running time of each operation increases by
$\O\brk*{\log_w(n)}$ additively.
\end{restatable}

\paragraph*{De-amortizing previous results.}
This lemma is general and can be applied to previous constructions.
For example, we describe how to de-amortize~\cite{kuszmaul2026succinct}.
The construction of~\cite{kuszmaul2026succinct},
called a compressed tabulation-weighted treap,
is a dynamic $\RANK/\SELECT$ data-structure
for $N$ elements,
using $\O(w)$ redundancy
and taking $\O_{\text{exp}}(\log^3{N})$ time for the operations.
Applying~\Cref{lemma:memory} to this tabulation-weighted treap for $N=2^{2(\log_w(n)^{1/3}}$,
yields a dynamic $\RANK/\SELECT$ data-structure with $\O(n+nw/\sqrt{N})$ redundancy,
and $\O_{\text{exp}}(\log_w(n))$ time operations.
When $N$ doubles (resp. halves), we lazily merge neighboring leaves (resp. split leaves), increasing the space by $\O(n)$ bits during this lazy procedure.

\subsubsection{Indices and dictionaries}

At the leaves, we construct a FID for $n=w^{\O(1)}$ elements, using $\lg\binom{u}{n}+nw^{\Th(1)}$ bits.
Our FID uses a lookup table of only $\polylog(w)$ words.
Following the ideas of~\cite{grossi2009more,patrascu2014dynamic,cohen2015minimal},
we divide the problem of constructing a small dynamic FID into implementing three data-structures:
\begin{enumerate}
\item A $\SELECT_1$ dictionary encodes a set $S\subseteq[u]$ and supports $\SELECT_1$ and update (i.e., insertion/deletion of elements of $S$).
\item A {$\RANK$ index}
is given $\O(1)$ time access to the $\SELECT_1$ dictionary
(i.e., assuming $\SELECT_1$ takes $\O(1)$ time)
and supports $\RANK$ and updates
(i.e., updates the index upon insertion/deletion of elements to and from $S$).
\item A {$\SELECT_0$ index}
is given $\O(1)$ time access to the $\SELECT_1$ dictionary
and supports $\SELECT_0$ and updates.
\end{enumerate}

To our knowledge, the best deterministic $\RANK$ index is
the dynamic fusion-node of~\cite{patrascu2014dynamic}.
It is a $\RANK$ index for
at most $w^{1/4}$ elements using $\Th(w)$ bits
and taking $\O(1)$ time for $\RANK$ and update.
We define the following notation for a $\RANK$ index.

\begin{definition}
For parameters $N$, $r$, and $t$,
a $\RANK$ index for a set of $n\leq N$ elements,
using $\O(nr)$ bits and taking $\O(t)$ time for $\RANK$ and update,
is called a $\tpl{N, r, t}$-index.
\end{definition}

We describe a $\SELECT_1$ dictionary
for $n=w^{\O(1)}$ elements
with $\O(\log_w(n))=\O(1)$ time operations
and $\lg\binom{u}{n}+\O(n\log w)$ bits.
Furthermore, we show how to use a given $\RANK$ index to implement a $\SELECT_0$ index, with only $\O(n\log w)$ additional bits (c.f.~\Cref{lemma:select0}).
As in~\Cref{lemma:memory}, we reduce the problem of $\RANK$ indices to just $w$ elements.
We show the following central lemma.
\begin{restatable}{lemma}{LemmaSmallFID}\label{lemma:small-FID}
Given a $\tpl{w, r, t}$-index,
there is a dynamic FID for a set of $n=w^{\O(1)}$ elements,
with $\O(nr+n\log w)$ redundancy,
taking $\O(t)$ time for $\RANK/\SELECT_0/$updates,
and $\O(1)$ time for $\SELECT_1$.
\end{restatable}

Lastly, we improve the $\RANK$ index of~\cite{patrascu2014dynamic},
by introducing a space-time trade-off with parameter $\eps$.
For $\eps=1/4$, this is the first (dynamic) $\RANK$ index
using $o(n\sqrt{w})$ bits with $\O(1)$ time operations,
answering this previously open question~\cite{kuszmaul2026succinct}.
The $\RANK$ index is as follows.
\begin{restatable}{theorem}{ThmSmallRank}\label{thm:small-rank}
For every $\eps\leq 1/2$,
there is a $\tpl{w, \O(w^\eps/\eps), \O(1/\eps)}$-index
in standard Word-RAM.
This index queries the $\SELECT_1$ dictionary
$\O(\log(1/\eps))$ times.
\end{restatable}

Putting these results together,
we prove~\Cref{thm:dynamic}.

\begin{proof}[Proof of~\Cref{thm:dynamic}]
Applying~\Cref{lemma:small-FID} together with~\Cref{lemma:memory},
with $N=w^2$ and
given a $\tpl{w,r,t}$-index,
yields a dynamic FID
occupying $\lg\binom{u}{n}+\O(nr+n\log w)$ bits,
taking $\O(t+\log_w(n))$ time for $\RANK/\SELECT_0$ and updates
and $\O(\log_w(n))$ time for $\SELECT_1$.
Using the index of~\Cref{thm:small-rank} with this result
concludes the proof,
observing $w^{\eps}/\eps=\Omega(\log w)$ for any $\eps$.
\end{proof}

The proof of~\Cref{thm:small-rank} consists of two non-trivial steps.
First, we modify the $\RANK$ index of~\cite{patrascu2014dynamic} to use less space,
yielding a $\tpl{w^{\eps},\O(w^\eps/\eps), \O(1/\eps)}$-index (\Cref{coro:index}).
Second, we build the index of~\Cref{thm:small-rank}
using an exponential tree~\cite{andersson2007dynamic},
where the leaves use this index
and each node in level $i$ uses the $\tpl{w^{\eps_i},\O(w^{\eps_i}/\eps_i), \O(1/\eps_i)}$-index where $\eps_i=1/2^{i+1}$.

For $\eps=1/\sqrt{\lg w}$,
we exploit binary search to improve the space.
This yields a $\RANK$ index that uses $o(n)$ space and $o(\log n)$ time for $n=\O(w)$ elements.

\begin{restatable}{theorem}{SqrtLog}\label{thm:sqrt-log}
There is a $\tpl{w, 1/2^{\O(\sqrt{\log w})}, \O(\sqrt{\log w})}$-index
in standard Word-RAM.
This index queries the $\SELECT_1$ dictionary $\O(\sqrt{\log w})$ times.
\end{restatable}


\subsection{Organization}

In~\Cref{sec:prelim}, we define the computational and memory model that we will use
and outline relevant techniques.
In~\Cref{sec:cardinality}, we describe the cardinality reduction and show~\Cref{lemma:memory}.
In~\Cref{sec:FID}, we show~\Cref{lemma:small-FID},
reducing the dynamic FID problem to a small $\RANK$ index.
In~\Cref{sec:small}, we show the small dynamic $\RANK$ index of~\Cref{thm:small-rank} in the Word-RAM with multiplication,
which proves~\Cref{thm:dynamic},
and further show~\Cref{thm:sqrt-log},
In~\Cref{sec:extension},
we discuss other Word-RAM models,
and analyze the static case.
In~\Cref{sec:conclusion}, we summarize our results and discuss open problems.


\section{Preliminaries}\label{sec:prelim}

\paragraph*{Terminology.}
For a FID, we only consider updates that add or remove an element from the set $S$.
For a sequence, updates that change its length are called \emph{indels}:
insertion of an element between index $i-1$ and $i$, so it becomes the new element at index $i$,
and deletion of the element at index $i$, so that the index $i+1$ becomes the element at index $i$.
A write to an entry of the sequence can be implemented as a deletion followed by an insertion.

\paragraph*{Notation.}
Let $[u]=\brc{0,\cdots, u - 1}$.
Denote the size of a set $S$ by $\size{S}$.
We denote an ordered set by $\brc{a_0<a_1<\cdots<a_{n-1}}$.
Denote by $\lg n\triangleq\ceil{\log_2 n}$.
We denote the concatenation of $x$ and $y$ by juxtaposition $xy$ and denote $x$ repeated $n$ times by $x^n$.
Our notation always writes the most-significant bit first.
For an array $A$,
we denote the range $A[i], A[i+1],\cdots, A[j]$ by $A[i:j]$.
In a tree, the degree of node $\nu$, i.e., the number of children,
is denoted $\deg(\nu)$.


\subsection{Computational model}

We employ the Word-RAM model, where the memory is an array of $2^w$ words of $w$-bits.
Memory reads and writes of a single word take $\O(1)$ time.
We assume $w\geq \lg u$.

\paragraph*{Instruction set.}
The standard \AC operations\footnote{An \AC circuit has constant depth and polynomial size, using NOT gates at the input and unbounded fan-in ANDs and ORs~\cite{andersson1999fusion}.}
of a $w$-bit word take $\O(1)$ time, namely
addition, bit-shifts ($\ll$ and $\gg$),
and bit-wise operations: AND ($\wedge$), OR ($\vee$), XOR ($\oplus$), and NOT ($\neg$)~\cite{fredman1990blasting,fredman1993surpassing,andersson1999fusion}.
We consider the ``standard'' Word-RAM model that has constant-time integer multiplication~\cite{fredman1990blasting,hagerup2001deterministic,patrascu2014dynamic}.
Other operations,
such as
most-significant bit ($\MSB$),
least-significant bit ($\LSB$),
Hamming-weight/popcount ($\WT$),
and integer division by a constant,
can all be implemented in terms of multiplication~\cite{fredman1990blasting,granlund1994division,patrascu2014dynamic}.
We note that operations can be restricted to $b$ bits by computing $x\mapsto x\wedge 0^{w-b}1^b$ after every operation.
In~\Cref{sec:other}, we consider other instruction sets.

\paragraph*{Constants.}
We allow access to at most $\polylog(w)$ words of precomputed constants that depend only on $w$ (e.g., $(0^{\sqrt{w}-1}1)^{\sqrt{w}}$), without accounting for it in the space, as long as each word is computable in $\O(w)$ time~\cite{ruvzic2008uniform,patrascu2014dynamic}.
A common use for this is to store a lookup table,
e.g., an array encoding a function $f:[\polylog(w)]\to\brc{0,1}^w$
(e.g., $f(i)=(0^{2^i-1}1)^{w/2^i}$).

\subsubsection{Packed arrays}\label{sec:bit}

Often, we reinterpret a word as an array of $k\leq w/b$ blocks of $b$ bits,
which we call a $k\times b$ packed array.
We denote the $i$-th block by $x\bitf{i}{b}=x[i\times b: i\times b+b-1]$.
When $b$ is implied, we omit $b$ and write $x[i]$.

The usefulness of a packed array of $\O(w)$ bits
resides in the fact that
many array operations (which usually take linear time)
can be implemented on it
in $\O(1)$ time in standard Word-RAM~\cite{fredman1990blasting,fredman1993surpassing,patrascu2014dynamic}.
Namely,
the following operations can be computed $\O(1)$ time:
\begin{enumerate*}[label=(\roman*)]
\item
element-wise operations $z\bitf{i}{b}=x\bitf{i}{b}\diamond y\bitf{i}{b}$ for $\diamond\in\brc{+,\wedge,\vee,\oplus}$,
\item
broadcast operations
$z\bitf{i}{b}=x\bitf{i}{b}\diamond \alpha$ for
a constant $\alpha$ and
$\diamond\in\brc{\ll, \gg, \times}$,
and
\item
parallel comparisons $z\bitf{i}{b}=0^{b-1}1$ iff $x\bitf{i}{b}\leq y\bitf{i}{b}$ and $0^b$ otherwise.
\end{enumerate*}
Some care is required for operations that can overflow (see the proof of~\Cref{lemma:parallel}).
Using $\LSB$,
we extend these to perform
element search (find $i$ such that $x\bitf{i}{b}=e$)
and,
if the elements are sorted, compute $\RANK$,
all in $\O(1)$ time.

We are also able to apply an operation to only a few entries if we have a mask for those entries.
For example,
range indexing $x\bitf{i:j}{b}$
and comparison indexing $x\bitf{i}{b}$ where $x\bitf{i}{b}\leq e$.
Concretely, for a function $q:[k]\to\brc{0,1}$,
consider a conditional write,
$x\bitf{i}{b}\gets v\bitf{i}{b}$ if $q(i)=1$, and $x\bitf{i}{b}\gets x\bitf{i}{b}$ otherwise.
If we have the mask
$M\bitf{i}{b}=q(i)^b$, the conditional write
is implemented by
$x\gets (x\wedge \neg M)\oplus (v\wedge M)$.
More commonly,
we have a mask $M'\bitf{i}{b}=0^{b-1} q(i)$,
but it can be converted in $\O(1)$ time by
$M\bitf{i}{b}=\neg(\neg M'\bitf{i}{b}+1)$, or equivalently, $M=M'\times 1^b$.

With these techniques, we support indels
on a packed array in $\O(1)$ time.
We cite the following trick used in~\cite{patrascu2014dynamic}
to extend this idea to
support indels for a small sequence.

\begin{lemma}[\cite{patrascu2014dynamic}]\label{lemma:indel}
For a sequence $A$ of length $k=\O(w/\log w)$,
we support indels into $A$ in $\O(1)$ time using $\O(k\log w)$ additional bits.
\end{lemma}
\begin{proof}
We store $A$ out of order, denoted $A^*$,
and a $k\times \O(\log w)$ packed array $P$
such that $A[i]=A^*[P[i]]$.
Note that $P$ fits in $\O(1)$ words, thus indels in $P$ are $\O(1)$ time.
An insertion of $x$ at position $i$ adds $x$ to the end of $A^*$ and inserts its index at position $i$ in $P$ (i.e., $P[i+1:k]\gets P[i:k-1]$ and $P[i]\gets k$).
A deletion at position $i$
moves the element $A^*[k-1]$ to position $P[i]$,
updates at position $i^*$ where $P[i^*]=k-1$ into $P[i^*]\gets i$
and
deletes at position $i$ in $P$ (i.e., $P[i:k-2]\gets P[i+1:k-1]$).
Note $i^*$ can be found in $\O(1)$ time by element search.
\end{proof}

\subsubsection{Multiply-shift}

Another relevant technique
is the multiply-and-shift~\cite{patrascu2014dynamic}.
For $x,\mu\in\brc{0,1}^{kb}$ and $s=\brc{s_i}_{i=0}^{k-1}$,
consider the following expression:
\begin{equation}
\Lambda_{\mu, s}(x)=\sumstyle
\sum_{i=0}^{k-1}\brk!{(x\wedge \mu)\bitf{i}{b} \ll s_i}
\end{equation}
We require
$\brk{\mu\bitf{i}{b}\ll s_i}\wedge\brk{\mu\bitf{j}{b}\ll s_j}=0$,
for every $i\neq j$,
so that the terms of the sum do not overlap.
If the mask $\mu$ and the shift values $s$ obey certain further conditions,
we will show how to implement $\Lambda_{\mu, s}$ in $\O(1)$ time in standard Word-RAM.
The idea is that multiple left shifts
that do not overlap
can be implemented by
a single multiplication with
a precomputed value.
We claim the following lemma.

\begin{lemma}\label{lemma:multiply-shift}
Let $\delta\geq\max_i (i\times b-s_i)$,
so that $s_i+\delta\geq i\times b$ for all $i$.
If for every $i\neq j$,
$\brk{\mu \ll (s_i+\delta-i\times b)}\wedge \brk{\mu \ll (s_j+\delta-j\times b)} = 0$,
then $\Lambda_{\mu, s}$ can be computed in $\O(1)$ time.
\end{lemma}
\begin{proof}
Note that this condition
implies that the terms of $\Lambda_{\mu, s}$ do not interfere, i.e.,
$\brk{\mu\bitf{i}{b}\ll s_i}\wedge\brk{\mu\bitf{j}{b}\ll s_j}=0$.
Let $\delta_i=s_i+\delta$.
Given this, we perform the following transformations to $\Lambda_{\mu, s}(x)$.
\begin{align}
\sumstyle
\sum_{i=0}^{k-1}\brk!{(x\wedge \mu)\bitf{i}{b} \ll \delta_i}
&=\sumstyle\sum_{i=0}^{k-1}\brk*{\brk!{\brk{(x\wedge\mu) \gg (i\times b)} \wedge \mu\bitf{i}{b} } \ll \delta_i} \notag\\
&=\sumstyle\sum_{i=0}^{k-1}
\brk*{ \brk!{(x\wedge\mu) \ll (\delta_i-i\times b)} \wedge \brk!{\mu\bitf{i}{b} \ll \delta_i} } \notag\\
&=\sumstyle
\brk*{\sum_{i=0}^{k-1}
\brk!{(x\wedge\mu) \ll (\delta_i-i\times b)} }
\wedge
\brk*{\sum_{i=0}^{k-1}
\brk!{\mu\bitf{i}{b} \ll \delta_i} }~,
\end{align}
where the last equality follows by the required condition on $\mu$.
Let $\alpha_0=\sum_{i=0}^{k-1}(1\ll \brk{\delta_i-i\times b})$
and $\alpha_1=\sum_{i=0}^{k-1}(\mu\bitf{i}{b}\ll \delta_i)$,
and therefore $\Lambda_{\mu, s}(x)=\brk{((x\wedge\mu)\times \alpha_0)\wedge \alpha_1}\gg \delta$.
\end{proof}

This construction uses values that are precomputed from the mask and the sequence of shifts.
In~\Cref{sec:small}, we will need to update these values.
For example, in~\Cref{lemma:multiply-shift},
modifying a single $s_i$ that maintains the condition on $\mu$ and $s$,
then
$\alpha_0$ and $\alpha_1$ can be updated in $\O(1)$ time.


\subsection{Memory model}\label{sec:mb}
We define the space used by a data-structure to be the highest address it uses in the memory array of words.
That is,
even if we do not use some intermediate cells,
the last used cell is what matters in the memory usage.
This is the $\mathcal{M}_B$ model of~\cite{raman2003succinct}.
To measure the space in bits, we account for the number of bits used in the highest-addressed word.
See~\Cref{fig:memory}.
All results in this paper are in $\mathcal{M}_B$.

\begin{figure}[ht]
\centering
\begin{tikzpicture}[scale=0.4]
\draw  (1/2, 1) node[anchor=south] {\small $0$};
\draw  (18+1/2, 1) node[anchor=south] {\small $2^w-1$};
\foreach \i in {0,1,3,5}
    \draw[fill=lightgray]  (\i, 0) rectangle ++(1, 1);
\foreach \i in {2,4}
    \draw  (\i, 0) rectangle ++(1, 1);
\draw[fill=lightgray]  (6, 0) rectangle ++(2, 1) node[midway] {$\cdots$};
\foreach \i in {0,1}
    \draw[fill=lightgray]  (8+\i, 0) rectangle ++(1, 1);
\draw[fill=lightgray,stroke=none] (10, 0) rectangle ++(0.5, 1);
\foreach \i in {0,1}
    \draw  (10+\i, 0) rectangle ++(1, 1);
\draw  (12, 0) rectangle ++(2, 1) node[midway] {$\cdots$};
\foreach \i in {0,1,..., 4}
    \draw  (14+\i, 0) rectangle ++(1, 1);
\draw[<->] (10+1/2, -1/4) -- (0, -1/4);
\draw (2, -2) node[anchor=west] {
$\begin{aligned}
    s\text{ bits } &=
    \floor{s/w} \text{ cells }\\
    &+ \text{partial word of } (s - w\floor{s/w}) \text{ bits }
\end{aligned}$};
\end{tikzpicture}
\caption{
Depiction of memory and space usage in the $\mathcal{M}_B$ model.
Cells in use are marked gray.
}
\label{fig:memory}
\end{figure}
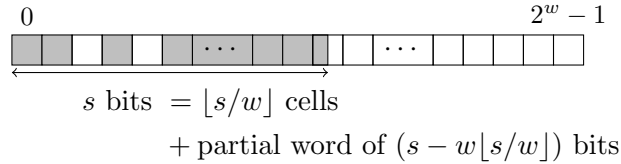

The main difficulty of this model is concatenating two or more dynamic data-structures.
This issue is studied in detail in~\cite{li2023dynamic}.
Next, we provide lemmas that will suffice for most of our uses.
A more delicate approach is necessary to prove~\Cref{lemma:memory}.

\subsubsection{Linear-space structures}\label{sec:linear}

In this section,
we describe procedures for concatenating data-structures
increasing the space by a constant, which we call linear space.
Here, a \emph{data-structure} (an array, a FID, a dictionary, a tree, etc.)
has its own view of an array of cells in $\mathcal{M}_B$.
We assume that the data-structure operations
allocate/de-allocate $\O(w)$ bits per operation.

We present the following lemma to append a linear data-structure to an existing structure.
This lemma can be applied repeatedly to concatenate at most $\O(1)$ linear data-structures.

\begin{lemma}\label{lemma:linear}
Given
a data-structure $\DD_1$ that uses $s_1$ bits
and another data-structure $\DD_2$ that uses $s_2$ bits.
Then, we can concatenate them to use $s_1+\O(s_2)$ bits,
increasing $\O(1)$ time for operations.
\end{lemma}
\begin{proof}
Re-parameterize so that $\DD_1$ and $\DD_2$ allocate/de-allocate at most $w$ bits.
We store $\DD_1$, then a free space of $\O(s_2)$ bits, then $\DD_2$,
and at the end another space of $\O(s_2)$ bits.
In $s_2/(2w)$ operations,
$\DD_1$ and $\DD_2$ have new sizes $s_1'$ and $s_2'$, respectively.
Because the sizes change by at most $w$ bits per operation,
we have
$s_1-s_2/2\leq s_1'\leq s_1+s_2/2$
and $s_2/2\leq s_2'\leq 3s_2/2$.
We pad $\DD_2$ with $s_2/2$ bits to allow it to grow as needed.
There is enough time to keep a copy of $\DD_2$
in the free space and in the end space.
If the free space has more than $2s_2$ bits, we keep the free space copy and discard the others,
and if the free space has less than $s_2/2$ bits, we keep the end space copy.
Note that the concatenation allocates/de-allocates $\O(1)$ words for every word that each structure allocates.
\end{proof}

Lastly,
we consider data-structures $\DD_n$ with a parameter $n$,
with operations that change $n$ by at most $1$.
We outline the standard trick
to avoid rebuilding
a data-structure when $n$ changes.
This idea is also useful for de-amortization,
especially for data-structures that have
computed values that strongly depend on $n$
(e.g.~\Cref{thm:compress}).
The idea is to allocate space for $\DD_{2^k}$ for $0\leq k\leq 2+\lg n$.
The data-structure in use is stored in $k=1+\lg n$.
We keep a copy of the data-structure in each $k=2+\lg n$ and $k=\lg n$ (only a partial copy of the smaller).
If $n$ doubles, we start allocating the next power-of-two,
and if $n$ halves, we de-allocate the last one.
We summarize with this observation.

\begin{observation}\label{obs:linear}
Suppose that, for any parameter $N$,
$\DD_n$ uses $\O(Ns)$ bits if $n=\Th(N)$.
Then, we can de-amortize $\DD_n$ to use $\O(ns)$ bits for any $n$,
adding $\O(1)$ time to the operations.
\end{observation}


\section{Cardinality reduction}\label{sec:cardinality}

In this section, we show a general procedure for reducing
the FID problem to small sets.
This is often referred to as \emph{cardinality reduction} (see \citep[{\S}3.2]{navarro2020predecessor} for details).
The key idea
is to use a \emph{leaf-oriented tree} \cite{dietz1989optimal,willard2000examining,andersson2007dynamic,bille2018dynamic}
where all leaves are at the same distance to the root,
and to employ auxiliary data-structures in the internal nodes
to reduce the operations of $\RANK$ and $\SELECT$ to a single leaf.


\subsection{Partial-sums}\label{sec:auxiliary}

The main data-structure we will need to navigate the leaf-oriented tree
is called a \emph{searchable partial-sum} (SPS) \cite{fredman1982complexity,dietz1989optimal,patrascu2004lower,patrascu2006logarithmic,arroyuelo2016succinct,bille2018dynamic}.
An SPS encodes a sequence $A\subseteq[\sigma]^k$, supporting the following operations.
\begin{itemize}
\item $\SUM(A, i)=\sum_{j=0}^{i}A[j]$, and
\item $\SEARCH(A, z)=\max\set{i\in[k]}{z\geq \SUM(A, i)}$ ($-1$ if empty).
\end{itemize}
For a given parameter $1\leq \delta\leq \lg\sigma$, a dynamic SPS supports updates of the form $A[i]\gets A[i]\pm \Delta$ for $\Delta\in[2^\delta]$.
We will also support insertions/deletions in $Z[i]=\SUM(A, i)$.
That is, a deletion at $Z[i]$ replaces $A[i]$ and $A[i+1]$ with $A[i]+A[i+1]$,
and an insertion of $z$ at $Z[i]$ replaces $A[i]$ with $z-\SUM(A, i)$ and $\SUM(A, i+1)-z=A[i]-\brk{z-\SUM(A, i)}$.
The result of~\cite{bille2018dynamic} states that there is an SPS
using $\O(kw)$ bits
with $\O(\log_{w/\delta}(k))$ operations.
We shall describe this data-structure in detail and refine the space bound in~\Cref{sec:partial-sum}.


\subsection{Leaf-oriented trees}\label{sec:navigation}

In a leaf-oriented tree,  all leaves are at the same
level (i.e., distance to the root),
and all nodes in the same level have around the same degree, up to a constant factor.
Concretely, there is a sequence of integers $\brc{d_i}_{i=0}^{\ell-1}$ such that
$d_i\leq \deg(\nu)\leq 4d_i$ for every internal node $\nu$ at level $i$ (except perhaps the root).
These are a generalization of
$B^+$-trees~\cite{comer1979ubiquitous},
$(a,b)$-trees~\cite{huddleston1982new},
$y$-fast trees~\cite{willard1983log},
and exponential trees~\cite{andersson2007dynamic}.
For example, a $(d, 4d)$-tree has $d_i=d$,
and an exponential tree has $d_i=d^{1/2^i}$.
The internal nodes store data-structures that are used to navigate to a leaf.

The general idea is to reduce an array $A$ to a smaller size.
The internal nodes of the tree maintain
a set of ``splitters''
$\brc{z_i}_{i=0}^{\Lambda-1}$ that divide the array.
The leaf $\lambda$ represents the array $A[z_\lambda:z_{\lambda+1}-1]$,
where $z_{-1}=0$ and $z_{\Lambda}$ is the length of $A$.
For traversing the tree, each node maintains its own smaller set of splitters
for the smaller array that subtree represents.
Formally, for a node $\nu$ with degree $d=\deg(\nu)$
it divides the array at positions $\brc{z_i}_{i=0}^{d-1}$.
Each node stores auxiliary data-structures
to reach the correct leaf.
For example,
to index the array $A$ at position $j$,
we use an SPS for an array $R$
where $\SUM(R, i)=z_i$,
and
we compute
$A[j]=(A[z_{i}:z_{i+1}-1])[j-z_i]$,
where $i=\SEARCH(R, j)$.

When inserting/deleting a child into/from a node $\nu$,
if the degree of $\nu$ becomes too large, we split it into two even parts,
and if it becomes too small,
we merge with a sibling into either a single node or two nodes (updating the parent of $\nu$ accordingly).
We show in~\Cref{sec:de-amortization}
how to lazily perform merges and splits,
using $\O(1)$ updates to the data-structures in node $\nu$
to get worst-case update times.
For pointers, we support merging and splitting
as indels in an array pointers,
using~\Cref{lemma:indel}
(we always assume $d_i=\O(w/\log w)$).
Usually, we also employ merges/splits to the leaves
to keep the size of the array between $N$ and $4N$,
for a given parameter $N$,
so that the number of splitters is $\Lambda=\Th(n/N)$.

We apply this structure to reduce the FID to a small set.
In this case, we consider the bit-vector $A[0:u-1]$,
and the splitters divide the universe.
However, instead of every array $A[z_\lambda:z_{\lambda+1}-1]$ having a similar size,
we want them to have a similar weight $\WT(A[z_\lambda:z_{\lambda+1}-1])=\Th(N)$.
The leaf $\lambda$ stores the set
$S_\lambda=\set{y-z_\lambda}{y\in S,~z_\lambda\leq y < z_{\lambda+1}}
\subseteq[z_{\lambda+1}-z_{\lambda}]$ with $z_\Lambda=u$.
Employing merges/splits,
we guarantee $N\leq \size{S_\lambda}\leq 4N$.
We emphasize that $z_\lambda$ are not necessarily in $S$;
this is important to keep the universe size fixed (before merging/splitting).

We detail how to compute $\RANK$ and $\SELECT$.
Consider a node with degree $d$
that divides the set $S$ at positions $\brc{z_i}_{i=0}^{d-1}$.
The $i$-th subtree
represents the set $S_i=\set{y-z_i}{y\in S,~z_i\leq y < z_{i+1}}$.
The node stores the elements $\brc{z_i}_{i=0}^{d-1}$
sorted, supporting indels with~\Cref{lemma:indel}.
We store two searchable partial-sums $N_b$ in the node,
where $\SUM(N_b, i)=\RANK_b(S, z_i)$
($\delta=1$ and $\sigma\geq n$).
Updates to $S$ are $\pm 1$ updates to the partial-sums.
To compute $\SELECT_b$ for a rank $r$, let $i=\SEARCH(N_b, r)$
and $r_i=\RANK_b(S, z_i)=\SUM(N_b, i)$.
Hence, we reduce the problem to $S_i$ as follows.
\begin{align}
\SELECT_b(S, r)
&=\SELECT_b\brk*{\set{y\in S}{z_i\leq y < z_{i+1}},~r - r_i} 
=z_i+\SELECT_b\brk*{S_i,~r-r_i}~.
\end{align}
For $\RANK$,
we store in the node a dynamic $\RANK$ index for the set of splitters
($\SELECT_1$ is $\O(1)$ time by accessing the sorted array).
Merges and splits, respectively, delete and insert splitters,
updating the $\RANK$ index.
To compute the $\RANK$ of $x\in[u]$,
let $i=\RANK(\brc{z_i}_{i=0}^{d-1}, x)-1$, therefore:
\begin{align}
\RANK(S, x)
&=\RANK(S, z_i)+\RANK\brk*{\set{y\in S}{z_i\leq y < z_{i+1}},~x} \label{eq:rank-a}\\
&=\RANK(S, z_i)+\RANK\brk*{S_i,~x-z_i}~.\label{eq:rank-b}
\end{align}

Using the SPS of~\cite{bille2018dynamic} and
the $\RANK$ index of~\cite{patrascu2014dynamic},
the internal nodes have degree $d_i=w^{1/4}$
and use $\O(d_i w)$ bits.
Therefore, the internal nodes use $\O(\Lambda w)=\O(nw/N)$ bits
and take $\O(\log_w(n))$ time to navigate to a leaf.


\subsection{De-amortization}\label{sec:de-amortization}

Let us describe the de-amortization procedures
for merging and splitting nodes.
We show how to perform merges and splits in a lazy fashion to get worst-case update time.
These are standard techniques: see~\cite{willard2000examining,andersson2007dynamic}.
Further, we show that lazily merging/splitting two FIDs uses only $\O(n)$ additional bits.
The goal is that the lazy operations have completed the merge/split
by the time that any of the merged/split sets later decide to further merge/split.

A node holds between $N$ and $4N$ elements.
The lazy operations will be completed in $N$ steps.
The parent node stores the intermediate state of the lazy operations.
Note that in the interim, there are at most $N$ insertions/deletions.
We describe the procedure for general leaf-oriented trees,
working with internal nodes or leaves,
where leaves can represent arrays or FIDs.

If the node reaches $4N$ elements,
we split into two nodes $\nu_1$ and $\nu_2$,
each holding $2N$ elements.
In detail, $\nu$ is the $\nu_1$ node, and we create a new $\nu_2$ node.
We move the last $5$ elements from $\nu$ to $\nu_2$,
so that within at most $(N+4N)/5=N$ operations, the procedure is completed.
For merging two nodes $\nu_1$ and $\nu_2$, with $n_1\leq 2N$ and $n_2\leq 2N$ elements respectively,
create a new node $\nu$.
We move the last $3$ elements from $\nu_1$ and
the first $3$ elements from $\nu_2$ to $\nu$,
so that within at most $(N+4N)/6\leq N$ operations, the procedure is completed.
If $n_1\geq 2N$, we combine both merging and splitting into two nodes of at most $3N$ elements.

The space for the interior nodes is still $\O(nw/N)$ bits.
For the leaves representing FIDs, the space used when splitting is bounded by
$\log_2\binom{u_1}{t}
+\log_2\binom{u_2}{t}
+\log_2\binom{u_1+u_2}{n-2t}
\leq \log_2\binom{2(u_1+u_2)}{n}
=\log_2\binom{u_1+u_2}{n}+\O(n)$ bits,
where we use the following inequalities:
$\sum_i \log_2\binom{u_i}{n_i}\leq \log_2\binom{\sum_i u_i}{\sum_i n_i}$ and
$\log_2\binom{2u}{n}\leq \log_2\binom{u}{n}+\O(n)$.
When merging, the space is
$\log_2\binom{u_1}{n_1-t}
+\log_2\binom{u_2}{n_2-t}
+\log_2\binom{u_1+u_2}{2t}
\leq
\log_2\binom{2(u_1+u_2)}{n_1+n_2}=
\log_2\binom{u_1+u_2}{n_1+n_2}+\O(n_1+n_2)$ bits.

With this, update times are guaranteed worst-case.
Recall that each term also has some redundancy,
but we may simply increase the constant in the redundancy of the underlying FID.
Therefore, an additional $\O(n)$ space is incurred by
the de-amortization.


\subsection{Managing memory}\label{sec:memory}

We describe how to implement a leaf-oriented tree in $\mathcal{M}_B$.
We will separately manage
the internal nodes and
the leaves, that dominate the space.

For the internal nodes,
we use the de-amortization of linear structures as described in~\Cref{obs:linear}
to maintain linear space and worst-case times for arbitrary $n$.
For $N\leq n\leq 4N$, for some parameter $N$,
we pre-allocate the space for each level with a large enough constant
so that the nodes in each level can be allocated as needed,
then concatenate them.
This technique works even if
we use $\O(\log w)$-bit pointers,
say, for a tree that fits in $w^{\O(1)}$ bits.

We now prove~\Cref{lemma:memory}.
Conceptually, each leaf FID has a view of its own $\mathcal{M}_B$ array.
The challenge is to
provide a view of linear memory to each leaf FID
all in a single array and
without incurring a lot of additional space.
The memory management we use is similar to the techniques used in~\cite{raman2003succinct,li2023dynamic,kuszmaul2026succinct},
with the main difference being that we exploit the leaf-oriented tree
to allocate/de-allocate memory per leaf.
This seemingly small change guarantees worst-case times
in our case where the number of leaves can change.

\LemmaMemory*
\begin{proof}
We use a leaf-oriented tree where each leaf $\lambda$ handles a set $S_\lambda$ of $\Th(N)$ elements.
Let $\brc{z_\lambda}_{\lambda=0}^{\Lambda-1}$ be the splitters, as described in~\Cref{sec:navigation},
let $u_\lambda=z_{\lambda+1}-z_\lambda$,  where $z_{-1}=0$ and $z_{\Lambda}=u$, and let $n_\lambda=\size{S_\lambda}$.
The space used by the leaves is:
\begin{equation}
\sumstyle
 \sum_{\lambda=0}^{\Lambda-1} \brk*{\lg\binom{u_\lambda}{n_\lambda}+\O(r)}\leq \lg\binom{\sum_\lambda u_\lambda}{\sum_\lambda n_\lambda}+\O(\Lambda r)
 = \lg\binom{u}{n}+\O(nr/N)~\text{bits.}
\end{equation}
As shown in the previous sections,
the leaves incur additional $\sum_\lambda \O(n_\lambda+r)=\O(n+nr/N)$ bits to the redundancy,
considering merging/splitting nodes.

For managing the memory of the FIDs,
we use chunks of $s=\Th(w\sqrt{N})$ bits; they are stored first in the memory array.
There are at most $\brk{\lg\binom{u_\lambda}{n_\lambda}+\O(r)}/s=\O(Nw)/s=\O(\sqrt{N})$ chunks per leaf $\lambda$.
For each leaf $\lambda$, we keep an array of pointers $P_\lambda$ to its chunks.
This array takes $\O(\sqrt{N}w)$ bits, which we store in an additional chunk.
Accessing an address $a$ in the FID of leaf $\lambda$ translates into
reading address $P_\lambda[\floor{a/s}]+(a\mod s)$.
We round $s$ to a power-of-two to simplify this expression.
Each leaf wastes $\O(1)$ chunks and $\O(w)$ bits.

We cannot use~\Cref{lemma:linear} directly
because an operation may need to allocate a new chunk of $\O(\sqrt{N}w)$ bits.
We exploit the fact that when a leaf allocates a new chunk,
it will take $\O(\sqrt{N})$ operations until it asks for another (and similarly, until it frees another chunk).
After $t$ insertions, the space changes by at most
$\O(tr)$ bits for the redundancy and
$\log_2\binom{u}{n+t}-\log_2\binom{u}{n}=\log_2\binom{u-n}{t}-\log_2\binom{n+t}{t}\leq \log_2\binom{u}{t} \leq t\lg u$ bits
for the main term.
Hence, a leaf may need a new chunk
only after $s/(\lg u+\O(r))\geq s/\O(w)=\Omega(\sqrt{N})$ insertions.
Deletions are analogous.

Leveraging this observation,
after the chunks, we maintain a free space of $\O(nw/\sqrt{N})$ bits,
and only after that, the internal nodes ($\O(nw/N)$ bits) using~\Cref{lemma:linear}.
When a new chunk is allocated in the free space,
we allocate an additional word at every operation.
There is room for $\O(1)$ unused chunks for each leaf
in the free space,
and hence allocation of a chunk takes $\O(1)$ time.
At the cost of adding $\O(w)$ bits to $s$,
we keep a free-list of unused chunks
and lazily move the last used chunk in the free space
to the oldest unused chunk in the free-list.
This adds $\O(1)$ chunks per leaf,
and hence this FID uses
$\lg\binom{u}{n}+\O(n+nr/N)+
\O(w\sqrt{N})\cdot \O(n/N)=\lg\binom{u}{n}+\O(n+nr/N+nw/\sqrt{N})$ bits.
\end{proof}


\section{Small FIDs}\label{sec:FID}

In this section, we construct a dynamic FID given only a $\RANK$ index.
First, we show how to construct a $\SELECT_0$ index using almost the same redundancy as the $\RANK$ index.
Next, we show a dynamic $\SELECT_1$ dictionary for $n=w^{\O(1)}$ elements using $\lg\binom{u}{n}+\O(n\log w)$ bits.
Together, this proves~\Cref{lemma:small-FID}.


\subsection{Succinct partial-sum}\label{sec:partial-sum}

Before we construct our small FID,
we first describe a searchable partial-sum
that will be used in the next sections.
The details of this construction will be useful to implement the $\SELECT_0$ index.
We provide a version of the construction of~\cite{patrascu2004tight,patrascu2006logarithmic},
with the adaptions of \cite{bille2018dynamic} to support indels in the sum.
We assume that $\sigma$ is fixed with $\lg\sigma\leq w$, and that we are allowed $\O(1)$ words of constants that depend on $\lg\sigma$ and $w$.
We claim the following theorem.

\begin{theorem}[Based on {\citep[{\S}8]{patrascu2006logarithmic}} and {\citep[{\S}3]{bille2018dynamic}}]\label{thm:partial-sum}
For $k=\O(w/(\delta+\log w))$,
given a $\tpl{w, r, t}$-index,
there is a searchable partial-sum for a sequence $A\in[\sigma]^k$
using $k\lg\sigma+\O(kr+k\log w)$ bits, with $\O(1)$ time $\SUM$
and $\O(t)$ time $\SEARCH$, updates, and indels in the sum.
\end{theorem}
\begin{proof}
Let $b=2^{\lg\lg(2^{\delta+2}k^2)}\leq 2\lg(2^{\delta+2}k^2)=\O(\delta+\log w)$.
We store a $k\times b$ packed array $T$
and an array $D$ of $k$ $\lg(\sigma w)$-bit elements ($\lg(\sigma w)=\lg\sigma+\O(\log w)$),
such that $\SUM(A, i)=D[i]+T[i]$.

To aid with the $\SEARCH$ operation, we define the set of representative positions of the array $A$.
Let $\brc{r_0<r_1<\cdots< r_{m-1}}\subseteq [k]$ such that $A[r_i]\geq 2^{\delta+1}k$,
and store the set in a packed array $R$ of $k\times \lg w$ bits.
The array $D$ maintains $D[r_i]=D[r_i+1]=\cdots=D[r_{i+1}-1]$.
This implies that the sequence $T[r_i:r_{i+1}-1]$ is non-decreasing.
We guarantee $2^\delta k\leq T[r_i]\leq 2^{\delta+1} k$ for every position $r_i\in R$; which implies $2^\delta k\leq T[i]\leq 2^{\delta+1}k^2+2^{\delta}k$ for any $i$.
Use the given $\RANK$ index for $D[R]$, i.e., the set $\brc{D[r_0], \cdots, D[r_{m-1}]}$.
For $\SEARCH(A, z)$, let $i=\RANK(D[R], z-1)$, i.e.,
$D[r_i]< z\leq D[r_{i+1}]$.
Observe $z>D[r_i]=\SUM(A, r_i)-T[r_i]\geq \SUM(A, r_i)- 2^{\delta+1}k\geq \SUM(A, r_{i-1})$,
thus we do not need to check indices below $r_{i-1}$.
If $z=\SUM(A, r_{i+1})$, or $z=\SUM(A, r_i)$, we are done,
otherwise, if $z<\SUM(A, r_i)$, let $j=i-1$,
and, if $z\geq \SUM(A, r_i)$, let $j=i$;
return $r_j+\RANK(T[r_j:r_{j+1}-1], z-D[r_j])$.

All updates are handled through $T$ in $\O(1)$ time:
for $A[i]\gets A[i]\pm \Delta$, compute $T\gets T\pm((0^{b-1}1)^k\ll (i\times b))\times\Delta$.
We want to maintain $2^\delta k\leq T[i]\leq 2^{\delta+1}k^2+2^{\delta}k$ (unless $D[i]=0$), so that after $k$ operations, it does not overflow or underflow.
To rebuild lazily, we keep an index $t$ and update $D[t]$ and $T[t]$, advancing $t$;
this is repeated $4$ times per update.
We keep two copies of $R$ and the index of $D[R]$, for $A[0:t-1]$ and $A[t:k-1]$.
If $A[t]\geq 2^{\delta+1}k$, set $T[t]=\min\brc{D[t]+T[t], 2^\delta k}$ and $D[t]=\max\brc{D[t]+T[t]-2^\delta k, 0}$, and update $R$ and $D[R]$ accordingly.
Otherwise, $D[t]=D[t-1]$ and $T[t]\gets D[t]+T[t]-D[t-1]$, and update each $R$ and $D[R]$.
Note that $R$ can be updated in $\O(1)$ time
using a conditional write,
when adding and removing representatives.
Updating the $\RANK$ index takes an additional $\O(t)$ time.

Indels in the sum can be supported as follows.
Because $k=\O(w/\log w)$,
we implement indels in $D$ as before,
storing it out of order and
using a permutation, represented as a packed array of $k\times \lg w$ bits.
After an indel, at most two elements of $A$ change,
and thus at most two representative elements need updating (to insert or delete from $D[R]$).
We update $T$, $R$, and the $\RANK$ index of $D[R]$ analogously to partial-sum updates.
We use~\Cref{lemma:linear}
to concatenate $D$ with the other structures.
\end{proof}

\begin{corollary}\label{coro:partial-sum}
For $k=\O(w/(\log\sigma+\log w))$,
there is an SPS that uses $k\lg\sigma+\O(k\log w)$ bits
and all operations take $\O(1)$ time.
\end{corollary}
\begin{proof}
We only need the packed array $T$ of $k\times\lg(\sigma w)$ bits,
where $T[i]=\SUM(A,i)$.
Therefore, $\SEARCH(A, z)=\RANK(T, z)$
and as $T$ fits in a word,
all operations are $\O(1)$ time.
\end{proof}

The important observation that we will need
in the next section is the following.
If the array $Z[i]=\SUM(A, i)$ is already stored and up-to-date,
then, we do not need to store $D$ and instead can rewrite all reads of $D[i]$ into $Z[i]-T[i]$,
and ignore all updates to $D$.


\subsection{Extending indices}\label{sec:select0}

Given a $\tpl{w, r, t}$-index for $\RANK_1$,
it still remains to solve $\SELECT_0$ for a small set $S$ of $n=w^{\O(1)}$ elements.
Using a leaf-oriented tree,
we reduce a set of $n=w^{\O(1)}$ element
to $\O(w/\lg w)$ elements,
using an additional $\O(nw/(w/\log w))=\O(n\log w)$ bits
and $\O(1)$ time.
Consider $n=\O(w/\lg w)$.

We use the following idea, first presented in~{\citep[{\S}2]{grossi2009more}} for the static case.
We keep an SPS for an array $M$ such that $\SUM(M, i)=\SELECT_1(S, i)-i$,
encoded with~\Cref{thm:partial-sum} using the given $\RANK$ index.
Then, $\SELECT_0(S, i)=i+\SEARCH(M, i)$ and
an update of $S$ is an indel to the sum of $M$, followed by an update.
That is,
to insert $x_i$ into $S$, we insert (in the sense of indel) $x_i$ into the sum of $M$.
However,
after this insertion,
all ranks after $i$ must decrease by $1$, which we fix with an update of $\Delta=-1$ into $M$.
Deletion is analogous.

We already have access to a $\SELECT_1$ dictionary for $S$ that is up-to-date.
In the notation used in the proof of~\Cref{thm:partial-sum},
we have an array $D$ with
$D[i]=\SELECT_1(S, i)-i-T[i]$.
The crucial point is that we do not need to store $D$ at all,
as computing $D[i]$ takes $\O(1)$ time using the $\SELECT_1$ dictionary and $T$.
Thus, we only need $\O(nr+n\log w)$ additional bits
to compute $\SELECT_0(S, i)$.
We summarize with the following lemma.

\begin{lemma}\label{lemma:select0}
Given a $\tpl{w, r, t}$-index
there is a $\SELECT_0$ index for a set of $n=w^{\O(1)}$ elements
using $\O(nr+n\log w)$ bits and
taking $\O(t)$ time for $\SELECT_0$ and update.
\end{lemma}


\subsection{Constructing the dictionary}\label{sec:dict}

Next,
we implement a $\SELECT_1$ dictionary using $\lg\binom{u}{n}+\O(n\log w)$ bits, for a set of $n=w^{\O(1)}$ elements.
We also describe how to manage the memory of this dictionary,
together with $\RANK$ and $\SELECT_0$ indices,
to prove~\Cref{lemma:small-FID}.

\LemmaSmallFID*
\begin{proof}
The operations $\RANK/\SELECT_0$ and update are as in~\Cref{sec:select0}, using the given $\RANK$ index;
this takes $\O(nr+n\log w)$ bits and $\O(t)$ time for $n=\O(w)$.
For $n=w^{\O(1)}$, we use a leaf-oriented tree,
increasing the time by a constant factor.
The internal nodes have degree $w$ and use the same $\RANK$ index,
further incurring $\O(nw/w)=\O(n)$ bits.

Denote $S=\brc{x_0<x_1<\cdots< x_{n-1}}\subseteq[u]$.
To build the $\SELECT_1$ dictionary,
we store the array $V=\brc{x_i}_{i=0}^{n-1}$, using $n\lg u$ bits.
Indels are supported in $V$ as follows, extending~\Cref{lemma:indel} to $n=w^{\O(1)}$.
We store $V$ out of order, denoted $V^*$,
with only append/remove at the end,
and store an array $P$ such that $V[i]=V^*[P[i]]$.
We use a leaf-oriented tree with
leaves partitioning $P$
into $\Th(w/\lg w)$ elements.
As mentioned, a packed array representing $\Th(w/\lg w)$ elements
of $P$ can be maintained
using $\O(w)$ bits and $\O(1)$ time for operations.
Similarly, internal nodes have degree $\Th(w/\lg w)$
and employ the construction of~\Cref{coro:partial-sum} for navigation.
To support deletion, we store the maximum of each subtree
for every internal node.
Updating position $i^*$ where $P[i^*]=n-1$ simply requires walking down the tree following the maximum,
which can be computed and updated in $\O(1)$ time per node.
The leaves of $P$ take $\O(n\log w)$ bits and,
accounting for the pointers, the internal nodes for $P$ incur another $\O(nw/(w/\log w))=\O(n\log w)$ bits

We use~\Cref{lemma:linear} to concatenate $V^*$ with $P$ and the $\RANK$ and $\SELECT_0$ indices.
The space used is
$n\lg{u}+\O(nr+n\log w)
=\lg\binom{u}{n}+\O(n\log n)+\O(nr+n\log w)
=\lg\binom{u}{n}+\O(nr+n\log w)$
bits, recalling that $n=w^{\O(1)}$.
\end{proof}

\begin{corollary}
For $n=w^{\O(1)}$,
there is a $\SELECT_1$ dictionary
using $\lg\binom{u}{n}+\O(n\log w)$ bits
with $\O(1)$ time for all operations.
\end{corollary}


\section{Small indices}\label{sec:small}

In this section, we first describe a $\RANK$ index that stores $n\leq w^{\eps/2}$ using $\O(nw^{\eps}/\eps)$ bits of space with $\O(1/\eps)$ time for $\RANK$ and update in standard Word-RAM,
for any given parameter $\eps\leq 1/2$.
That is, we construct a $\tpl{w^{\eps/2}, \O(w^\eps/\eps), \O(1/\eps)}$-index.
This is a modification based on the index given in~\cite{patrascu2014dynamic} to use fewer bits.
We then extend this index to $n=\O(w)$, keeping the space and time complexity, proving~\Cref{thm:small-rank}.
By exploiting binary search for $\eps=1/\sqrt{\lg w}$, we prove~\Cref{thm:sqrt-log}.


\subsection{Compressed keys}

We present a systematization and generalization of the construction of~\cite{patrascu2014dynamic}.
The main idea within a small index is to first compute $\RANK$ on a compressed key, which uses only the most significant bits that differentiate elements in the set.
There are at most $n$ significant bits for a set of $n$ elements.
We do not allow for unused zero bits in the compressed key, as in~\cite{patrascu2014dynamic}, differing from~\cite{fredman1993surpassing}.

\begin{definition}
Given a set $S$ of $n$ $w$-bit values,
define the set of significant bit-positions to be
$\brc{c_0<c_1<\cdots< c_{k-1}}=\set{\MSB(x\oplus y)}{x,y\in S}$.
Importantly, $k\leq n$.
For $x\in\brc{0,1}^w$, a compressed key $\hat{x}\in\brc{0,1}^k$ is a $k$-bit value such that $\hat{x}[i]=x[c_i]$, for every $i\in[k]$.
Equivalently, we define $\hat{x}=\sum_{i=0}^{k-1}(x[c_i]\ll i)$.
\end{definition}

We represent the set of significant bits by a packed array $C$.
For an array $C$ of $n$ bit positions,
denote $\Pack{C}(x)=\sum_{i=0}^{n-1}(x[C[i]]\ll i)$.
A \emph{compression} is a procedure that computes $\Pi_C(x)$,
and thus, computes a compressed key.
An update to $S$ will insert or delete at most a single significant bit-position,
and thus updates are indels into $C$.
We remark that
there is a non-standard \AC operation that is input $C$  and $x$ and computes $\Pack{C}(x)$~\cite{andersson1999fusion}.
However,
this is not possible in the standard Word-RAM model,
and thus we will need to use ancillary values, similar to~\Cref{lemma:multiply-shift}.
We carefully consider that a compression must update those values according to the updates to $C$.

Given a compression,
it suffices to construct a $\RANK$ index over the set of compressed keys $\set{\hat{x}}{x\in S}$.
In~\cite{cohen2015minimal},
it is argued that any such index
can be used to solve the weak prefix problem~\cite{belazzougui2010fast},
which was shown to require $\Omega(n\log w)$ bits.
We cite the following theorem that implements this index using a packed array while supporting updates.

\begin{theorem}[{\citep[{\S}3]{patrascu2014dynamic}}]\label{thm:rank-node}
Suppose we are given a compression for a set of at most $n=\O(\sqrt{w})$ elements using $\O(s)$ ancillary bits and taking $\O(t)$ time for compression and update.
Then, there is a $\RANK$ index using $\O(n^2+n\log w+s)$ bits
and taking $\O(t)$ time for $\RANK$ and update,
and accessing the $\SELECT_1$ dictionary $\O(1)$ times.
\end{theorem}
\begin{proof}[Proof Sketch]
Assume that $n$ is a power-of-two.
Let $\hat{x}=\Pack{C}(x)$.
For $i\in[n]$, let $\varphi_i\in\brc{0,1}^n$ be such that $\varphi_i[j]=1$ iff $c_j=\MSB(x_i\oplus z)$ for some $z\in S\setminus\brc{x_i}$; and $0$ otherwise.
Store an $n\times\lg w$ packed array $C$ with $C[i]=c_i$,
and two masks $B,F\in\brc{0,1}^{n^2}$ where
$F\bitf{i}{n}=\neg\varphi_i$ and
$B\bitf{i}{n}=\hat{x}_i\wedge \varphi_i$.
Define $\MATCH(x)=\RANK(B\vee (\hat{x}^n \wedge F),~\hat{x})$,
which takes $\O(t)$ time, due to the compression.

For $\RANK(S, x)$,
we compute $i=\MATCH(x)$ and $j=\MSB(x_i\oplus x)$.
To get $x_i$, we access the $\SELECT_1$ dictionary.
Then, $\RANK(S, x)$ is:
$i$ if $x=x_i$,~
$\MATCH(x\wedge (1^{w-j}0^{j}))-1$ if $x<x_i$, and
$\MATCH(x\vee (0^{w-j}1^{j}))$ if $x>x_i$.
Updates to $B$ and $F$ also take $\O(t)$ time, as one must first match the inserted/deleted element.
\end{proof}

The compression of~\cite{patrascu2014dynamic} for $n\leq w^{1/4}$ stored $\Th(w)$ ancillary bits that depend on $C$, which dominates the $\O(n^2+n\log w)$ term.
If we insist on storing fewer bits, we must change the compression, as even a mask of the significant bits takes $w$ bits.
As in~\cite{fredman1993surpassing,patrascu2014dynamic}, we will use multiplication as the main workhorse for our compression.

Let us describe a version of the compression given in~\citep[{\S}3]{patrascu2014dynamic}
that works for a $b$-bit value and uses only $b$-bit operations.
First, we show that we can permute $n$ bits arbitrarily,
as long as they are already packed into $\O(w/n)$ bits,
using the  same techniques as~\Cref{lemma:multiply-shift}.
We will use this result again in the next section.

\begin{lemma}\label{lemma:permute}
Given $b$ and $n$
where $n\leq b$ and $2bn\leq w$,
and an array $C$ of $n$ bit positions,
the function
$\Pack{C}:\brc{0,1}^b\to\brc{0,1}^n$
can be computed in $\O(1)$ time
using $\O(b)$ ancillary bits that depend on $C$.
\end{lemma}
\begin{proof}
Let $\mu\in\brc{0,1}^b$ where $\mu[C[i]]=1$ for $i\in[n]$; $0$ otherwise.
Define the following functions on $b$-bit values:
$\Lambda_0(x)=\sum_{i=0}^{n-1} \brk!{(x\wedge \mu)[C[i]]\ll (i\times 2b)}$, and
$\Lambda_1(x)=\sum_{i=0}^{n-1} \brk!{(x\wedge (0^{2b-1}1)^{n})\bitf{i}{2b}\ll i}$,
so that $\Pack{C}(x)=\Lambda_1(\Lambda_0(x))$.
We use~\Cref{lemma:multiply-shift} for $\Lambda_1$,
and for $\Lambda_0$ we follow a similar procedure to get
$\Lambda_0(x)=\brk*{\brk*{(x\wedge \mu)\times \alpha} \gg b}\wedge (0^{2b-1}1)^{n}$,
where $\alpha=\sum_{i=0}^{n-1} \brk{1\ll (i\times 2b+b-C[i])}\in\brc{0,1}^{2bn}$.
\end{proof}

The final trick is to show that one can pack $n$ significant bits into $\O(n^2)$ bits
in $\O(1)$ time.
We provide the full proof here for completeness
and to make explicit some properties of the compression that we will need later.
This is a slight generalization
as it allows $C$ to be any array of bit positions,
not just representing a set.

\begin{theorem}[Generalization of {\citep[{\S}3]{patrascu2014dynamic}}]\label{thm:perfect}
For $b\le w$ and an
array $C$ of $n$ bit positions,
where $n\leq \min\brc{b^{1/3},w^{1/4}}$,
the function
$\Pack{C}:\brc{0,1}^b\to\brc{0,1}^n$
can be computed in $\O(1)$ time
using $\O(b)$ ancillary bits that depend on $C$.
An update to $C$ takes $\O(1)$ time
to update the ancillary bits.
\end{theorem}
\begin{proof}
Assume that $n=\Th(\min\brc{b^{1/3},w^{1/4}})$ and that $n$ is a power-of-two.
Due to rounding, we can divide $C$ into two halves
and compute the compression in each separately.
This increases the number of ancillary bits and the time by at most a constant.

Let $C^*$ be the unique and sorted bit positions of $C$.
We represent $C$ and $C^*$ by $n\times \lg b$ packed arrays,
and we maintain a permutation $P$ such that $C[i]=C^*[P[i]]$.
Thus, $\Pack{C}(x)=\Pack{P}(\Pack{C^*}(x))$,
where $\Pack{P}$ uses~\Cref{lemma:permute}.
Updates to $C$ take $\O(1)$ time to update $C^*$ and $P$
using packed array operations.
Concretely, to insert $c$ into $C$ at position $i$, we first compute $j=\RANK(C^*, c)$.
If $C^*[j]\neq c$, we insert $c$ into $C^*$ at position and
add $1$ to all indices $i^*$ with $P[i^*]\geq j$.
This computation takes $\O(1)$ time using a parallel comparison followed by a conditional write (c.f.~\Cref{sec:bit}).
Then, insert $j$ into $P$ at position $i$.
To update $\Pi_P$,
let $\alpha_P$ be the ancillary multiplier of $\Pi_P$ of~\Cref{lemma:permute},
i.e., $\alpha_P\bitf{i^*}{2n}=1\ll (n-P[i^*])$.
To account for the change in $P$, we shift by $1$ all entries
in $\alpha_P$ with $\alpha_P\bitf{i^*}{2n}\leq 1\ll (n-j)$
(using a parallel comparison and a conditional write),
then insert $1 \ll (n-j)$ at position $i$.
Deletion is analogous.
We now assume $C=C^*$.

Let $d=n^2/2$.
First, we prove this general procedure to pack blocks into $\O(d)$ bits without overlap.
Given $u\in\brc{0,1}^{d}$ and $v\in\brc{0,1}^{2d}$ where $\WT(u)+\WT(v)\leq n$,
there is a shift value $0\leq s\leq d/2$ such that $(u\ll s)\wedge v =0$.
Observe $\brk{(u\ll s)\wedge v}[i]=1$ iff $v[i]=1$ and $u[i-s]=1$.
For a pair $v[i]=u[j]=1$, the shift $s=i-j$ yields an overlap in position $i$.
Then, there are at most $\WT(u)\times\WT(v)\leq (n/2)^2=d/2$ shift values that yield overlaps,
and hence at least one $0\leq s\leq d/2$ yields no overlaps,
as there are $d/2+1$ options.
To compute $s$,
we compute $z=(u\times \alpha_S)\wedge (v\times (0^{2d-1}1)^{d})$,
where
$\alpha_S=\sum_{s=0}^{d/2}(1\ll (s+s\times 2d))\in\brc{0,1}^{2d^2}$.
This value fits in a word, as $2d^2=2n^4/4\leq w$.
A shift value $s$ is good iff $z\bitf{s}{2d}=0$,
and we find it in $\O(1)$ time
using element search in a packed array.

Let $\mu\in\brc{0,1}^b$ where $\mu[C[i]]=1$ for $i\in[n]$; $0$ otherwise.
Updates to $\mu$ also take $\O(1)$ time.
Let $k=\ceil{b/d}$ and let $\brc{s_i}_{i=0}^{k-1}$ be defined recursively as follows.
For every $i\in[k]$,
we apply the previous procedure
with $u=\mu\bitf{i}{d}$
and $v=\sum_{j=0}^{i-1}(\mu\bitf{j}{d}\ll s_j)$
to find the shift value $0\leq s_i\leq d/2$.
Define:
\begin{equation}
\Upsilon(x)=\sumstyle\sum_{i=0}^{k-1}((x\wedge \mu)\bitf{i}{d}\ll s_i)\in\brc{0,1}^{2d}~.
\end{equation}
We compute $\Upsilon$ with a variation of~\Cref{lemma:multiply-shift}.
Let $\alpha=\sum_{i=0}^{k-1}(1\ll (s_i+(k-1-i)\times d))\in\brc{0,1}^{b}$.
To avoid overlap between blocks in the multiplication,
we mask the even and odd blocks.
To this end, denote $x_e= x\wedge E$ and $x_o=(x\gg d)\wedge E$, where $E=(0^{d}1^{d})^{k}$.
Therefore, we compute as follows:\footnote{This expression is corrected from~\cite{patrascu2014dynamic} for even $n$, as the shifts corresponding to the even indices are in the odd positions of $\alpha$.}
$\Upsilon(x)=
{\brk{
\brk{\brk{x\wedge \mu}_e\times \alpha_o }
\oplus
\brk{\brk{x\wedge \mu}_o\times \alpha_e }
}\gg ((k-2)\times d)}$.
After each multiplication,
there may be overlap on all but the last block.
However, it does not overflow to the neighboring even block, as $\lg(n\times 0^d1^{d})\leq 3d/2$,
and thus the expression correctly computes $\Upsilon$.

The value $\Upsilon(x)\in\brc{0,1}^{2d}$ has the relevant bits in some order:
$\Upsilon(x)[D[i]]=x[C[i]]$.
Namely, $D[i]=C[i]\wedge 1^{\lg(d)}+s_{C[i]\gg \lg(d)}$.
We now permute the $n$ significant bits to the correct order.
That is, we compute $\Pack{D}(z)=\sum_{i=0}^{n-1}(z[D[i]]\ll i)$
so that  $\Pack{C}(x)=\Pack{D}(\Upsilon(x))$.
Because $2(n/2)\cdot 2d=n^3\leq b$, we apply~\Cref{lemma:permute} separately to $D[0:n/2-1]$ and $D[n/2:n-1]$,
then concatenate the resulting $(n/2)$-bit values.

We briefly discuss updates in $C$.
Recall that $C$ is sorted.
When adding a bit-position $c$ to $C$,
we compute the new shift value for the block $i=c\gg \lg(d)$
by repeating the packing procedure with $u=\mu\bitf{i}{d}$ and $v=\Upsilon(\mu)$.
Then, update $D$ and the multiplier of $\Upsilon$.
The ancillary values used in $\Pack{D}$,
namely its multiplier $\alpha_D$,
can be updated by shifting
the segments associated with
the range of indices contained in the $i$-th block.
We compute this range by applying $\RANK$ to the packed array $C$.
This step requires that $C$ is sorted so that there is
a contiguous range of bit positions for the $i$-th block.
Deletions just need to update $\Pack{D}$.
This concludes the compression.
\end{proof}


\subsection{Using fewer bits}

We present our compression.
Our idea is to divide the word into blocks of $b=\Th(nw^\eps)$ bits,
and show that we can compress all blocks in parallel
in $\O(1)$ time using $\O(b)$ ancillary bits.
After gathering the blocks using an instance of~\Cref{lemma:multiply-shift},
we get a packed array of $wn/b=\Th(w^{1-\eps})$ bits containing all significant bits (in some order).
After repeating $\O(1/\eps)$ times, all significant bits fit within a block
and thus can be sorted to get the compressed key (using~\Cref{lemma:permute}).
We define the following class of functions that we will show that can be parallelized.
That is, that can be applied to many elements in a packed array,
taking the same time.

\begin{definition}
For $b\leq w$,
a function $f:\brc{0,1}^b\to\brc{0,1}^b$
that can be computed in $\O(1)$ time
is \textbf{simple} if it is of the form:
\begin{equation}
    f(x)=\begin{cases}
        x&\\
        \alpha&\text{for }\alpha\in\brc{0,1}^b\\
        g(x)\times \alpha&\text{for $g$ simple, $\alpha\in\brc{0,1}^b$}\\
        g(x)\diamond \alpha&\text{for $g$ simple, $\alpha\in[b]$, $\diamond\in\brc{\ll,\gg}$}\\
        g_1(x)\diamond g_2(x)&\text{for $g_1, g_2$ simple, $\diamond\in\brc{+,\wedge, \vee, \oplus}$}~.
    \end{cases}
\end{equation}
All operations are restricted to $b$ bits.
Equivalently, a function is simple if it uses at most $\O(1)$ operations and all multiplications and shifts are by a constant.
\end{definition}

The storage of $f$ is defined by the space used by the ancillary values $\alpha$.
Recall that values that only depend on $w$ can be stored in lookup tables
so as not to take any storage.
We prove that we can parallelize a simple function without additional storage.

\begin{lemma}\label{lemma:parallel}
For $b\leq w$ and $k\leq w/b$,
let $f:\brc{0,1}^b\to\brc{0,1}^b$ be a simple function.
Then, there is a simple function $f^k:\brc{0,1}^{kb}\to\brc{0,1}^{kb}$,
where $f^k(x)\bitf{i}{b}=f(x\bitf{i}{b})$ for every $x\in\brc{0,1}^b$ and $i\in[k]$,
and that uses the same ancillary values as $f$.
\end{lemma}
\begin{proof}
We define $f$ recursively.
For a constant $f(x)=\alpha$, let $f^k(x)=\alpha \times (0^{b-1}1)^{k}=\alpha^k$, i.e., $\alpha$ repeated $k$ times.
Otherwise, we must handle possible overflow.
We define the na\"ive extension $f^{(k)}(x)=g^k(x)\diamond \alpha$ or $f^{(k)}(x)=g_1^k(x)\diamond g_2^k(x)$.
For $\diamond \in\brc{\wedge,\vee,\oplus}$, we take $f^k=f^{(k)}$.
The other operations yield at most $b$ additional bits of overflow, thus we mask out the even $b$-bit blocks.
Let $E=(0^b1^b)^{\ceil{k/2}}\wedge 1^{kb}$.
Define
$f^{k}(x)=(f^{(k)}(x\wedge E)\wedge E) \oplus ((f^{(k)}((x \gg b)\wedge E)\wedge E) \ll b)$.
This procedure increases the number of operations by at most a constant factor, keeping time $\O(1)$.
Note that $E$ and $(0^{b-1}1)^{k}$ only depend on $b$ and $k$ and are stored in a lookup table.
\end{proof}

Note that~\Cref{lemma:multiply-shift} shows that some compressions can be implemented by a simple function.
The main observation is the following.
\begin{observation}
The compression of~\Cref{thm:perfect} is a simple function.
\end{observation}
Note that we do not need the updates to the ancillary bits to be simple functions.
Applying~\Cref{lemma:parallel},
we compress each $b$-bit block separately
in parallel, using the same compression and storing only $\O(b)$ bits.
Now, we need to gather the blocks, which we state next.

\begin{lemma}\label{lemma:pack}
For any $a\geq 1$ and $k\leq w/(a n)$,
there exist a function $\Gamma:\brc{0,1}^{ka n}\to\brc{0,1}^{kn}$
and a permutation $\gamma:[k]\to [k]$
that can be computed in $\O(1)$ time in standard Word-RAM,
such that $\Gamma(x)\bitf{\gamma(i)}{n}=x\bitf{i}{a n}\wedge 1^n$, for $i\in[k]$.
\end{lemma}
\begin{proof}
Let $b=an$, we divide $x$ into $k$ $b$-bit blocks.
Assume that $k$ is a multiple of $a$, otherwise, divide into the first $a\times\floor{k/a}$  blocks and the remaining blocks (at most $a$) and handle each separately.
To gather all blocks,
we interlace super-blocks of $kn$ bits that contain $k/a$ $b$-bit blocks.
Concretely, we compute
$\Gamma(x)=\sum_{i=0}^{a-1} \brk*{{(x\wedge (0^{b-n}1^n)^k)\bitf{i}{kn}} \ll (i\times n)} $,
using~\Cref{lemma:multiply-shift}.
Hence,
the block at index $i=i_0+i_1\times k/a$ in $x$ is mapped to index $i_1+i_0\times a=i\times a-(k-1)\times i_1$ in $\Gamma(x)$,
and thus
the permutation is $\gamma(i)=ia-(k-1)\times\floor{ia/k}$.
Its inverse is $\gamma^{-1}(i)=i\times k/a-(k-1)\times\floor{i/a}$.
For $k<a$,
let $\gamma(i)=i$ and
compute $\Gamma(x)=\sum_{i=0}^{k-1} \brk*{{(x\wedge (0^{b-n}1^n)^k)\bitf{i}{an}} \ll (i\times n)} $ using~\Cref{lemma:multiply-shift}.
\end{proof}

For $b\geq n^3$, the big idea is to pack $w$ bits to $wn/b$ bits, in some order,
and repeat $t$ times until $w\cdot (n/b)^t\leq b$, which can be packed in the correct order to get the compressed key.
We have the following theorem.

\begin{theorem}\label{thm:compress}
For every $\eps\leq 1/2$,
for $n\leq w^{\eps/2}$,
there is a compression
taking $\O(1/\eps)$ time in standard Word-RAM for compression and update,
storing $\O(nw^\eps/\eps)$ ancillary bits.
\end{theorem}
\begin{proof}
First, we store an array $C$ of the significant bits, taking $n\times \lg w$ bits.
Updates to $C$ maintaining sorted order take $\O(1)$ time.

If $\eps\leq \lg\lg w/\lg w$, it holds that $w^{\eps/2}\leq \sqrt{\lg w}=\O(1/\eps)$.
In this case, for $n\leq w^{\eps/2}$ the compression is computed na\"ively: $\hat{x}=\sum_{i=0}^{n-1}(x[C[i]]\ll i)$;
and it takes $\O(n)=\O(1/\eps)$ time, requiring only to store and update the packed array $C$ of significant bit-positions.
Recall $w^\eps/\eps=\Omega(\log w)$, thus the upper bound holds.

Otherwise, define $b=nw^{\eps}$, assume $n$ and $b$ are powers-of-two.
To avoid rebuilding when $n$ changes,
we de-amortize as described in~\Cref{obs:linear}.
Observe $n^3\leq nw^\eps=b$, thus~\Cref{thm:perfect} applies.
We proceed as follows: each block of $b$ bits is compressed into $n$ bits in parallel
in the same way,
then all blocks are compressed to $w\cdot (n/b)$ bits.
Apply this $t=\O(\ceil{1/\eps})$ times, until we get $w\cdot (n/b)^t\leq b$ remaining bits.
We allow the compression in each step to permute the significant bits.
This will be important for updates, as we will see.
In the last step, we compress into $n$ bits and correct the ordering.

Let us consider the first iteration.
We maintain an array $R$ of $n$ distinct elements from $[b]$
such that $\brc{C[i]\mod b}_{i=0}^{n-1}\subseteq\brc{R[i]}_{i=0}^{n-1}$.
Using~\Cref{lemma:parallel}, we construct $\Pack{R}^{w/b}(x)$ that packs every block in parallel, using only $\O(b)$ ancillary bits.
For an update,
we will change at most
one entry of $R$,
which are reflected in the ancillary bits of $\Pack{R}$, as per~\Cref{thm:perfect}.

To complete the step,
let $\Gamma$ be the function from~\Cref{lemma:pack} for $k=w/b$ and $a=b/n=w^{\eps}$.
Then, $\Lambda(x)=\Gamma(\Pack{R}^{w/b}(x))$ packs $x$ into $wn/b=w^{1-\eps}$ bits,
preserving the significant bits.
Concretely,
the significant bit at position $c=R[i]+j\times b$,
moves to position $i+j\times b$ by $\Pack{R}^{w/b}$,
then
moves to position $i+\gamma(j)\times n$, where $\gamma$ is the permutation from~\Cref{lemma:pack} of $\Gamma$.
The resulting position is thus $c'=i+\gamma(j)\times n$
where $i$ is the index of $c\mod b$ in $R$
and $j=\floor{(c-R[i])/b}$, all computed in $\O(1)$ time.
The original significant bit positions $C$ are moved to some new set of positions $C'$.
Essentially, $\Lambda(x)$ computes $\sum_{i=0}^{n-1}\brk{x[C[i]]\ll C'[i]}$.

To insert or delete position $c$,
we update the current $C$ and $R$,
and we need to ensure that only $c'$ needs to be updated in $C'$.
To this end, we store a packed array $F\in\brc{0,1}^n$ of unused positions.
That is, $F[i]=1$ iff $R[i]\notin \brc{C[j]\mod b}_{j=0}^{n-1}$; otherwise, $F[i]=0$.
Deletion is lazy:
when deleting position $c$,
we set $F[i]=1$
only if there is exactly one entry in $C$
with remainder $R[i]=c\mod b$.
Namely, the entry corresponds to $c$, found by element search in a packed array.
When inserting position $c$,
if $c\mod b$ is not in $R$,
we find an unused position $j=\LSB(F)$,
and set $F[j]=0$ and $R[j]=c\mod b$.
Hence, we modify $R$
without changing the permutation of the other bits, and
therefore the next iteration $C'$ needs to insert/delete at most one bit.

Let $C_s$ and $R_s$ be the respective $C$ and $R$ at step $s\in[t]$,
where $C'_s=C_{s+1}$.
Each step stores $\O(b+n\log w+n)=\O(nw^{\eps})$ bits.
Define $\Lambda_s(x)=\Gamma_s(\Pack{R_s}^{w/b}(x))$, where $\Gamma_s$ is the function from~\Cref{lemma:pack} for $k=(w/b)\cdot (n/b)^s$ blocks and $a=b/n=w^\eps$.
As $\Gamma_s$ only depends on $\lg n$, $\lg b$, $s$, and $w$,
we may store its ancillary bits as a constant lookup table ($\O(\log^3 w)$ words).
Let $\tilde{x}=\Lambda_{t-1}\brk{\cdots\Lambda_1\brk{\Lambda_0(x) } \cdots }$.
We packed $x$ into $b$ bits: $\tilde{x}\in\brc{0,1}^b$.
At this point, $C_t$ encodes the positions of the significant bits after packing, i.e., $\tilde{x}[C_t[i]]=x[C[i]]$.
Finally, the compression is
$\hat{x}=\Pack{C_t}(\tilde{x})=\Pack{C_t}\brk{\Lambda_{t-1}\brk{\cdots\Lambda_1\brk{\Lambda_0(x) } \cdots }}$
using~\Cref{lemma:permute}.
\end{proof}

Applying~\Cref{thm:rank-node}
using the compression of~\Cref{thm:compress},
gives a $\tpl{w^{\eps/2}, \O(w^{\eps}/\eps), \O(1/\eps)}$-index,
as $\O(n^2+n\log w+nw^\eps/\eps)=\O(nw^{\eps/2}+nw^\eps/\eps)=\O(nw^\eps/\eps)$.

\subsection{Leveraging trees}\label{sec:tree}

Next,
we show how to use leaf-oriented trees to increase the number of elements
that a $\RANK$ index can store up to $w$ elements, with the same space per element.
First, we need to refine the construction of~\Cref{sec:cardinality}
to use less space for the internal nodes and manage memory.

Looking in detail in~\Cref{sec:navigation},
to compute $\RANK$,
the internal nodes store a set of splitters $\brc{z_i}_{i=0}^{d-1}$,
and stores a $\RANK$ index for it and an SPS $N_1$ where $\SUM(N_1, i)=\RANK(S, z_i)$.
We use~\Cref{coro:partial-sum} to implement the SPS with $\sigma=\O(w)$.
To avoid storing $z_i$, we require $z_i\in S$
and use~\Cref{eq:rank-a} instead of~\Cref{eq:rank-b}.
Thus, $z_i=\SELECT_1(S, \SUM(N_1, i))$.
This requires replacing $z_i$ if it is deleted from $S$,
say by the minimum of its corresponding subset.
The consequence of this change is that all $\RANK$ indices must be over the same universe $[u]$.
Finally, we use $\O(\log w)$-bit pointers,
and therefore, an internal node of degree $d$ for this tree
uses $\O(d\log w)$ bits,
plus the space for its $\RANK$ index.
We use the technique described~\Cref{sec:memory} to maintain linear space
even when using $\O(\log w)$-bit pointers.

Given these constructions, we now build leaf-oriented trees for $\RANK$.
An immediate application of this is to
use a two-level $(w^{\eps/2}, 4w^{\eps/2})$-tree
with
the $\tpl{w^{\eps/2}, \O(w^{\eps}/\eps), \O(1/\eps)}$-index
resulting from ~\Cref{thm:rank-node,thm:compress},
to get the following.
\begin{corollary}\label{coro:index}
For every $\eps\leq 1/2$,
there is a $\tpl{w^{\eps}, \O(w^{\eps}/\eps), \O(1/\eps)}$-index.
\end{corollary}

We are now able to prove~\Cref{thm:small-rank}.
For this construction, we use exponential trees~\cite{andersson2007dynamic},
where the degree decreases exponentially with the level, i.e., $d_{i+1}=\sqrt{d_i}$,
so that we achieve $\O(1/\eps)$ time.
Here, we exploit the fact that our index of~\Cref{coro:index} is parameterized.

\ThmSmallRank*
\begin{proof}
We construct a leaf-oriented exponential tree with $\O(\lg(1/\eps))$ levels.
We use leaves of degree $w^{\eps}$ using the $\tpl{w^{\eps}, \O(w^{\eps}/\eps), \O(1/\eps)}$-index of~\Cref{coro:index},
and use $\O(nw^{\eps}/\eps)$ bits and take $\O(1/\eps)$ time.
For $0\leq i\leq \log(1/\eps)-1$, let $\eps_i=1/2^{i+1}$.
We use internal nodes of degrees $d_i=w^{\eps_i}$,
where each level $i$ uses the $\tpl{w^{\eps_i}, \O(w^{\eps_i}/\eps_i), \O(1/\eps_i)}$-index.
There are $\O(n/w^{\eps})$ leaves (at level $\lg(1/\eps)$),
and by induction, level $i$ has $\O(n/w^{\eps_{i-1}})$ nodes.
Note that $n/\sqrt{w}=\O(\sqrt{w})$, and
hence the root can use the $\eps=1/2$ index.
Therefore, level $i$ uses $(n/w^{\eps_{i-1}})\cdot \O(w^{\eps_i}\cdot w^{\eps_i}/\eps_i)=\O(n/\eps_i)$ bits,
and thus the space used by internal nodes is
$\sum_{i=0}^{\log(1/\eps)-1} \O\brk{n/\eps_i}=\O\brk*{\sum_{i=0}^{\log(1/\eps)-1} n2^i}=\O(n/\eps)$.
The internal pointers and partial-sums take $\O((n/w^\eps)\log w)=\O(n/\eps)$ bits.
Lastly, the time to navigate the internal nodes is
$\sum_{i=0}^{\log(1/\eps)-1} \O\brk{1/\eps_i}=\O\brk*{\sum_{i=0}^{\log(1/\eps)-1} 2^i}=\O(1/\eps)$.
Note that this index accesses the $\SELECT_1$ dictionary $\O(\log(1/\eps))$ times,
$\O(1)$ times per level.
\end{proof}


\subsection{Super-constant time}\label{sec:binary-search}

We investigate $1/\eps=\omega(1)$
to achieve $o(n\log w)$ bits for the $\RANK$ index.
The main trick we use is the following.
For indices that take $\O(t)=\omega(1)$ time,
we repeat
a binary search in the $\SELECT_1$ dictionary
so that the tree only stores $n/2^{\O(t)}$ elements.
This modification does not change the time (still $\O(t)$), but reduces the space by a factor of $2^{\O(t)}$.

For $\eps=1/\sqrt{\lg w}$,
we can also replace the index on the leaves of~\Cref{thm:small-rank}, which dominates the space, 
with a binary search.
This uses $\O(1/\eps + \log w^\eps)=\O(\sqrt{\log w})$ time
and $\O(n/\eps)=\O(n\sqrt{\log w})$ bits.
Furthermore,
repeating the binary search in the leaves $\O(1)$ additional times yields the following.

\SqrtLog*

If we want to reduce the number of accesses to the $\SELECT_1$ dictionary,
as in~\cite{cohen2015minimal},
recall that for $\eps\leq \lg\lg w/\lg w$,
the construction of~\Cref{thm:compress}
does not require ancillary values to achieve $\O(1/\eps)$ time.
That is, it is
a $\tpl{w^\eps, \O(\log w), \O(1/\eps)}$-index.
Using an exponential tree as in~\Cref{thm:small-rank},
the space is $\O(n/\eps)+\O(n\log w)=\O(n\log w)$ bits,
and the time is $\O(1/\eps)$.
That is a $\tpl{w, \O(\log w), \O(\log w/\log\log w)}$-index.
This index queries the $\SELECT_1$ dictionary $\O(\log\log w)$ times.
Applying binary search only increasing $\O(\log \log w)$ accesses,
reduces the space by a factor of $2^{\O(\log\log w)}=(\log w)^{\O(1)}$.

\begin{corollary}\label{coro:log-rank}
There is a $\tpl*{w, 1/(\log w)^{\O(1)}, \O\brk*{\frac{\log w}{\log\log w}}}$-index
in standard Word-RAM.
This index queries the $\SELECT_1$ dictionary $\O(\log\log w)$ times.
\end{corollary}


\section{Extensions}\label{sec:extension}

In this section,
we extend our technique to other cases of interest,
namely, different instruction sets and the static case.


\begin{table*}[t]
\renewcommand{\arraystretch}{1}
\centering
\begin{tabular}{|c|c|c|c|c|c|c|}\hline
Reference & Dyn. & Model & $N$ & Space/$N$ & Time & Accesses \\\hline
\cite{ajtai1984hash} & \checkyes & Cell-probe  & $\frac{w}{\lg w}$ & $\O(\log w)$ & $\O(1)$ & $\O(1)$ \\\hline
\cite{fredman1993surpassing} & \checkno & Multiplication & $w^{1/6}$ & $\O(w^{5/6})$ & $\O(1)$ & $\O(1)$ \\\hline
\cite{belazzougui2010dynamic} & \checkyes & Multiplication & $-$ & $\O(\log w)$ & $\O_\text{exp}(\log w)$ & $\O_{\text{exp}}(1)$ \\\hline
\cite{cohen2015minimal} & \checkno & Standard \AC & $\frac{w}{\lg w}$ & $\O(\log w)$ & $\O(\log w)$ & $\O(1)$ \\\hline
\cite{patrascu2014dynamic} & \checkyes & Multiplication & $w^{1/4}$ & $\O(w^{3/4})$ & $\O(1)$ & $\O(1)$ \\\hline
\cite{patrascu2014dynamic}+\cite{andersson1999fusion} & \checkyes  & \cite{andersson1999fusion} \AC & $w^\eps$ & $\O(w^\eps)$ & $\O(1)$ & $\O(1)$ \\\hline
Ours (\Cref{coro:index}) & \checkyes  & Multiplication & $w^{\eps}$ & $\O(w^\eps/\eps)$ & $\O(1/\eps)$ & $\O(1)$ \\\hline
Ours (\Cref{thm:small-rank}) & \checkyes  & Multiplication & $w$ & $\O(w^\eps/\eps)$ & $\O(1/\eps)$ & $\O(\log(1/\eps))$ \\\hline
Ours (\Cref{thm:sqrt-log}) & \checkyes  & Multiplication & $w$ & $1/2^{\O(\sqrt{\log w})}$ & $\O\brk*{\sqrt{\log w}}$ & $\O(\sqrt{\log w})$ \\\hline
Ours (\Cref{coro:log-rank}) & \checkyes  & Multiplication & $w$ & $1/(\log w)^{\O(1)}$ & $\O\brk*{\frac{\log w}{\log\log w}}$ & $\O(\log\log w)$ \\\hline
Ours (\Cref{coro:ac0}) & \checkyes  & \cite{andersson1999fusion} \AC & $w$ & $\O(w^\eps+\log w)$ & $\O(\log(1/\eps))$ & $\O(\log(1/\eps))$ \\\hline
Ours (\Cref{coro:ac0}) & \checkyes  & \cite{andersson1999fusion} \AC & $w$ & $1/(\log w)^{\O(1)}$ & $\O(\log\log w)$ & $\O(\log\log w)$ \\\hline
\end{tabular}
\caption{Survey of rank indices.
Some rows use a parameter $\eps\leq 1/2$.
``Accesses'' is the number of accesses to $\SELECT_1$ dictionary.
``$N$'' is the maximum number of elements in the index.
Some indices, namely~\cite{fredman1993surpassing,cohen2015minimal,patrascu2014dynamic},
always take $\O(w)$ bits, regardless of $N$.
Note the time bounds of~\cite{belazzougui2010dynamic} are randomized.
}
\label{tbl:rank}
\end{table*}

\subsection{Other computational models}\label{sec:other}

The implementation of a $\RANK$ index heavily depends  on the available $\O(1)$ time instructions of the model.
All models we discuss allow for standard \AC operations.
In the ``cell-probe'' model, any operation in $\O(1)$ words takes $\O(1)$ time.
The \AC model of~\cite{andersson1999fusion} adds a set of non-standard \AC operations in
$\O(1)$ time.
See~\Cref{tbl:rank} for a list of results.

Our general construction in
\Cref{lemma:small-FID} is also applicable to Word-RAM models with different instruction sets.
Note that we will also need to replace the $\RANK$ index of~\cite{patrascu2014dynamic} in the internal nodes of leaf-oriented trees of~\Cref{lemma:memory}.

In the cell-probe model~\cite{ajtai1984hash}, where any computation on $\O(1)$ words takes $\O(1)$ time,
there is a $\tpl{w/\log w,\O(\log w), \O(1)}$-index, storing a compressed trie as described in~\cite{ajtai1984hash,patrascu2014dynamic,cohen2015minimal}.
Applying~\Cref{lemma:memory,lemma:small-FID} in this model,
we get a dynamic FID with $\O(n\log w)$ redundancy
that takes $\O\brk*{\log_w(n)}$ time for all operations.

In the non-standard \AC model of~\cite{andersson1999fusion},
a compression, that is, $\tty{compress}(x,C)=\Pack{C}(x)$, can be computed $\O(1)$ time.
Hence, there is a compression using $\O(n\log w)$ bits,
and thus there is a $\tpl{w^\eps,\O(w^\eps+\log w), \O(1)}$-index for any $\eps\leq1/2$, using~\Cref{thm:rank-node}.
In the proof of~\Cref{thm:partial-sum} and~\Cref{lemma:select0},
we use multiplication to replicate a $\O(\log w)$-bit value $k$ times,
but this \AC model includes
a $\tty{replicate}(x, k)=x^k$
operation in $\O(1)$ time.
Applying exponential trees
analogously to~\Cref{thm:small-rank}
to this parameterized index
yields the following.
For $\eps=\omega(1)$, we use binary search as in~\Cref{sec:binary-search}.

\begin{corollary}\label{coro:ac0}
In the 
\AC model of~\cite{andersson1999fusion},
there is a $\tpl{w,\O(w^\eps+\log w), \O(\log(1/\eps))}$-index.
For $\eps=\lg\lg w/\lg w$,
this is a $\tpl{w,\O(\log w), \O(\log\log w)}$-index.
Binary search
reduces the space by a factor of $2^{\O(\log\log w)}=(\log w)^{\O(1)}$,
which yields a $\tpl{w,1/(\log w)^{\O(1)}, \O(\log\log w)}$-index.
\end{corollary}

The ideas of~\Cref{lemma:memory,lemma:small-FID} also apply,
using the index with $\eps=1/4$ as a replacement for~\cite{patrascu2014dynamic}.
Thus, we are able to construct FIDs corresponding to these $\RANK$ indices.


\subsection{Static case}

In the static case,
there are several improvements we can make to our construction.
First, we cite the lower bound of~\cite{patrascu2006time,patrascu2014dynamic} for
the optimal time for predecessor search.
The predecessor over a set of $n$ elements using $\O(nw)$ space takes $\Omega(\OPT(u,n))$ time, where:
\begin{equation}\label{eq:opt}
\OPT(u,n)=
\min\brc*{
\log_w(n), \log\brk*{\log_w\brk*{\frac{u}{n}}},
\frac{\log\log_w(u)}{\log\brk*{\frac{\log\log_w(u)}{\log\log_w(n)}}}
}~.
\end{equation}
In~\cite{patrascu2006time}, there is a static structure achieving this bound.
For polynomial universes ($u=n^{1+\Th(1)}$),
this lower bound is $\OPT(u,n)=\Th(\log\log u)$,
and a static solution with optimal time and redundancy is known~\cite{liang2025optimal}.
We explore our construction for the remaining regimes.

We replace the internal nodes of the leaf-oriented tree (see~\Cref{sec:tree})
with this optimal predecessor data-structure for the splitters
$\brc{z_\lambda}_{\lambda=0}^{\Lambda-1}$,
where $\Lambda=\O(n/N)$,
using $\O(\Lambda w)$ bits and $\OPT(u, \Lambda)$ time
to navigate to the correct leaf for $\RANK$.
For $N=w^2$, we get $\OPT(u, \Lambda)=\Th(\OPT(u, n))$.
This predecessor data-structure
can also be used to implement $\SEARCH$ in a searchable partial-sum
to navigate to the correct leaf for $\SELECT_0$.

We also use the trick of~\cite{grossi2009more} in the static setting (which we adapt in~\Cref{sec:select0} for the dynamic setting)
to reduce $\SELECT_0$ to predecessor search.
In the static case, we can remove the $\O(n\log w)$ term in~\Cref{lemma:select0}.
Next,~\cite{grossi2009more} describes
a static $\SELECT_1$ dictionary,
for any set of $n$ elements,
using $\O(n)$ redundancy,
taking $\O(1)$ time for $\SELECT_1$.
Combining this with our $\RANK$ indices
yields the following.

\begin{corollary}[Static FID]\label{coro:static}
For every $\eps\leq 1/2$,
there is a static FID
with $\O(nw^{\eps}/\eps)$ redundancy,
taking $\O\brk*{1/\eps+\OPT(u,n)}$ time for $\RANK/\SELECT_0$,
and $\O(1)$ time for $\SELECT_1$.
For $\eps=1/\sqrt{\log w}$,
the redundancy is $\O(n)$.
\end{corollary}


\section{Conclusion and open problems}\label{sec:conclusion}

In this paper, we describe dynamic FID
that use $\lg\binom{u}{n}+\O(nw^{\eps}/\eps)$ bits.
Updates and $\RANK/\SELECT_0$ take $\O\brk*{1/\eps+\log_w(n)}$ time,
and $\SELECT_1$ takes $\O\brk*{\log_w(n)}$ time.
For $\eps=1/4$, this is the first deterministic dynamic FID
in the standard Word-RAM model
with $o(n\sqrt{w})$ bits of redundancy, optimal time for all operations, and worst-case times.

With our generic reductions of~\Cref{lemma:memory,lemma:small-FID},
further research may focus on building more space-efficient $\RANK$ indices.
Alternatively, different constructions of small FIDs such as~\cite{kuszmaul2026succinct} also fit into our framework, where we provide a general procedure to manage memory
with worst-case update times.
We cite a few remaining open questions.

\paragraph*{Lower bounds.}
The best lower bound is for the static case~\cite{patrascu2010cell}
that requires $n/w^{\O(t)}$ redundancy for $\O(t)$ time FID.
A recent lower bound~\cite{li2023tight} holds
for a polynomial universe ($u=n^{1+\Th(1)}$)
and $w=\Th(\log u)$,
and states that any dynamic dictionary
(that is, storing a dynamic set and supporting membership)
with $\O(n)$ redundancy
requires $\Omega(\log^*(w))$ time for membership,
where $\log^*$ is the iterated logarithm.
Considering~\Cref{lemma:memory},
a time lower bound for
a dynamic FID of $N=w^{\O(1)}$ elements
and $\O(n)$ redundancy
in the standard Word-RAM model
is an interesting open question.

\paragraph*{Indices and dictionaries.}
Our results of~\Cref{lemma:select0,lemma:small-FID} have an additive term of $\O(n\log w)$ bits
for the dynamic setting.
These are the dominating terms for FIDs using the $\RANK$ indices of~\Cref{thm:sqrt-log} and~\Cref{coro:log-rank,coro:ac0}.
For polynomial universes, there are some results achieving $\O(n)$ redundancy with $\omega(1)$ time $\SELECT_1$~\cite{pibiri2017dynamic,pibiri2020succinct}.
We leave for further work to improve these results.


\section*{Acknowledgments}
This research was supported by the Israel Science Foundation, grant No. 1948/21.
I wish to thank Guy Even and Boaz Patt-Shamir for their helpful comments and advice,
and the anonymous reviewers for their remarks.

\bibliography{main}

\end{document}